\documentclass[acmsmall,screen]{acmart}

\setcopyright{rightsretained}
\acmPrice{}
\acmDOI{10.1145/3276497}
\acmYear{2018}
\copyrightyear{2018}
\acmJournal{PACMPL}
\acmVolume{2}
\acmNumber{OOPSLA}
\acmArticle{127}
\acmMonth{11}

\usepackage{booktabs}   %% For formal tables:
\usepackage{subcaption} %% For complex figures with subfigures/subcaptions
\usepackage{enumitem}
\setlist
  { leftmargin = 5.5mm
  , itemsep = 3pt
  , topsep = 4pt
  }

\usepackage{mathpartir}
\usepackage{fancyvrb}
\usepackage{wrapfig}
\usepackage{cleveref}
\usepackage{fouridx}
\usepackage{amsmath}

\newcommand{\mynote}[3]{\textcolor{#3}{\textsf{{#2}}}}

\newcommand{\rkc}[1]{\mynote{rkc}{#1}{blue}}
\def\parahead#1{\paragraph{\textbf{#1}}}

\newcommand{\ie}{{{i.e.}}}
\newcommand{\eg}{{{e.g.}}}
\newcommand{\etc}{{{etc.}}}
\newcommand{\cf}{{{cf.}}}

\newcommand{\web}{web}
\newcommand{\Web}{Web}
\newcommand{\key}[1]{#1}

\newcommand{\littleLeo}{\ensuremath{\textsc{LittleLeo}}}
\newcommand{\leo}{\ensuremath{\textsc{Leo}}}
\newcommand{\sns}{\ensuremath{\textsc{Sketch-n-Sketch}}}
\newcommand{\snsLeo}{\ensuremath{\textsc{Sketch-n-Sketch}}}

\newcommand{\snsVersion}{\textsc{v0.7.1}}

\newcommand{\suppMaterials}{Supplementary Appendices}

\newcommand{\codeSize}{\small}
\newcommand{\codeSizeInFigure}{\codeSize}
\newcommand{\codeSizeInText}
  {\normalsize}
\newcommand{\toolName}[1]
  {#1}
\newcommand{\valuePropToolName}[1]
  {\toolName{Update Program}}
\newcommand{\browserToolName}[1]
  {#1}
\newcommand{\set}[1]{\{{#1}\}}

\newcommand{\sep}{\hspace{0.06in}}
\newcommand{\sepPremise}{\hspace{0.20in}}
\newcommand{\hsepRule}{\hspace{0.20in}}
\newcommand{\vsepRuleHeight}{0.12in}
\newcommand{\vsepRule}{\vspace{\vsepRuleHeight}}
\newcommand{\miniSepOne}{\hspace{0.01in}}
\newcommand{\miniSepTwo}{\hspace{0.02in}}
\newcommand{\miniSepThree}{\hspace{0.03in}}

\newcommand{\figSyntaxLineBreak}{\\[1pt]}
\newcommand{\figSyntaxSpaceNextCategory}{\\[1pt]}

\newcommand{\figSyntaxSpaceItem}{\sep\mid\sep}

\newcommand{\figSyntaxEnd}{\end{array}$}

\newcommand{\figSyntaxBegin}{$\begin{array}{rcl}}
\newcommand{\figSyntaxRowLabel}[2]{{\textbf{#1}}\ {#2}&\miniSepTwo{::=}\miniSepTwo&}
\newcommand{\figSyntaxRow}{&\miniSepOne\mid\miniSepOne&}

\newcommand{\relDescription}[1]{\ensuremath{\textrm{\textbf{#1}}}}
\newcommand{\judgementHeadCenter}[2]
  {\ensuremath{\fbox{\relDescription{#1}\hspace{0.15in}#2}}}

\newcommand{\ruleName}[1]{\mbox{\textsc{\begin{normalsize}#1\end{normalsize}}}}
\newcommand{\ruleNameFig}[1]{\textsc{\begin{scriptsize}[#1]\end{scriptsize}}}

\newcommand{\ttlcurly}{\ensuremath{\texttt{\char`\{}}\hspace{0.00in}}
\newcommand{\ttrcurly}{\ensuremath{\hspace{0.00in}\texttt{\char`\}}}}
\newcommand{\ttlparen}{\ensuremath{\texttt{(}}}
\newcommand{\ttrparen}{\ensuremath{\texttt{)}}}
\newcommand{\ttlbrack}{\ensuremath{\texttt{[}}}
\newcommand{\ttrbrack}{\ensuremath{\texttt{]}}}

\newcommand{\ttf}{\ensuremath{\texttt{f}}}

\newcommand{\ttm}{\ensuremath{\texttt{m}}}

\newcommand{\ttx}{\ensuremath{\texttt{x}}}

\newcommand{\ttz}{\ensuremath{\texttt{z}}}

\newcommand{\ttcomma}{\ensuremath{\texttt{,}}}

\newcommand{\varConst}{c}
\newcommand{\varNum}{n}

\newcommand{\varStr}{s}
\newcommand{\varBool}{b}
\newcommand{\varVal}{v}

\newcommand{\varExp}{e}
\newcommand{\varExpBare}{e}

\newcommand{\varVar}{x}
\newcommand{\varVarY}{y}
\newcommand{\varPat}{p}
\newcommand{\varPatBare}{p}

\newcommand{\varExpId}{i}

\newcommand{\varField}{f}
\newcommand{\varEnv}{E}

\newcommand{\helperop}[1]{\ensuremath{\mathit{#1}}}

\newcommand{\dom}[1]{\ensuremath{\helperop{dom}({#1})}}
\newcommand{\freeVars}[1]{\ensuremath{\helperop{fv}({#1})}}
\newcommand{\coerceToExp}[1]{\ensuremath{\helperop{exp}({#1})}}

\newcommand{\listDiff}
  {\ensuremath{\Delta}}
\newcommand{\diffKeep}
  {\ensuremath{\helperop{Keep}}}
\newcommand{\diffDelete}
  {\ensuremath{\helperop{Delete}}}
\newcommand{\diffInsert}[1]
  {\ensuremath{\helperop{Insert}({#1})}}
\newcommand{\diffUpdate}[1]
  {\ensuremath{\helperop{Update}({#1})}}
\newcommand{\emptyDiff}
  {\ensuremath{\valNil}}
\newcommand{\diffCons}[2]{\expCons{#1}{#2}}
\newcommand{\computeListDiff}[2]{\ensuremath{\helperop{Diff}({#1},{#2})}}
\newcommand{\foldListDiffEquals}[5]
  {\ensuremath{\updatesTo{#1}{#2}{{}_{\helperop{Diff}}\miniSepThree{#3}}{#4}{#5}}}
\newcommand{\coerceToVal}[1]{\ensuremath{\helperop{val}({#1})}}

\newcommand{\twoThings}[2]{\ensuremath{{#1}\miniSepThree{#2}}}
\newcommand{\threeThings}[3]{\ensuremath{{#1}\miniSepThree\twoThings{#2}{#3}}}
\newcommand{\fourThings}[4]{\ensuremath{{#1}\miniSepThree\threeThings{#2}{#3}{#4}}}
\newcommand{\fiveThings}[5]{\ensuremath{{#1}\miniSepThree\fourThings{#2}{#3}{#4}{#5}}}

\newcommand{\expOneThing}[2]{\ensuremath{{\twoThings{#1}{#2}}}}
\newcommand{\expTwoThings}[3]{\ensuremath{{\threeThings{#1}{#2}{#3}}}}
\newcommand{\expThreeThings}[4]{\ensuremath{{\fourThings{#1}{#2}{#3}{#4}}}}

\newcommand{\num}[1]{\ensuremath{\texttt{#1}}}
\newcommand{\op}[1]{\ensuremath{\texttt{#1}}}

\newcommand{\parens}[1]{\ensuremath{\ttlparen{#1}\ttrparen}}
\newcommand{\bracks}[1]{\ensuremath{\ttlbrack{#1}\ttrbrack}}
\newcommand{\spaceOne}[1]{\ensuremath{\miniSepOne{#1}\miniSepOne}}
\newcommand{\spaceTwo}[1]{\ensuremath{\miniSepTwo{#1}\miniSepTwo}}
\newcommand{\spaceThree}[1]{\ensuremath{\miniSepThree{#1}\miniSepThree}}
\newcommand{\expFun}[2]{\ensuremath{{\lambda}{#1}.\miniSepOne{#2}}}
\newcommand{\valNil}{\ensuremath{\ttlbrack\ttrbrack}}
\newcommand{\expCons}[2]{\ensuremath{{#1}\spaceTwo{::}{#2}}}
\newcommand{\expList}[3]{\bracks{\spaceOne{\threeThings{#1\ttcomma}{#2\ttcomma}{#3}}}}
\newcommand{\expListOne}[1]
  {\bracks{\spaceOne{#1}}}
\newcommand{\expListTwo}[2]
  {\bracks{\spaceOne{\twoThings{#1\ttcomma}{#2}}}}
\newcommand{\expListThree}[3]
  {\bracks{\spaceOne{\twoThings{#1\ttcomma}{#2\ttcomma}{\miniSepOne#3}}}}
\newcommand{\expListFive}[5]
  {\bracks{\spaceOne{\fiveThings{#1\ttcomma}{#2\ttcomma}{#3\ttcomma}{#4\ttcomma}{#5}}}}
\newcommand{\expApp}[2]{\ensuremath{{#1}\miniSepThree{}{#2}}}
\newcommand{\expAppTwo}[3]{\ensuremath{\expApp{#1}{\twoThings{#2}{#3}}}}

\newcommand{\expBinop}[3]
  {\ensuremath{{#2}\spaceThree{#1}{#3}}}
\newcommand{\expPlus}[2]
  {\expBinop{\op{+}}{#1}{#2}}
\newcommand{\expEqual}[2]
  {\expBinop{\op{==}}{#1}{#2}}

\newcommand{\expConcat}[2]
  {\expBinop{+}{#1}{#2}}
\newcommand{\charList}[1]
  {\ensuremath{\langle{#1}\rangle}}

\newcommand{\expStr}[1]
  {\ensuremath{\texttt{\textquotedbl{#1}\textquotedbl}}}
\newcommand{\expLongStr}[1]
  {\ensuremath{\expStr{\expStr{\expStr{#1}}}}}
\newcommand{\expOutputStr}[1]
  {\ensuremath{\textsl{\textquotedbl{#1}\textquotedbl}}}
\newcommand{\expTrue}{\ensuremath{\texttt{True}}}
\newcommand{\expFalse}{\ensuremath{\texttt{False}}}

\newcommand{\closureSyntax}[2]{\ensuremath{({#1},\miniSepThree{#2})}}
\newcommand{\closure}[3]{\ensuremath{\closureSyntax{#1}{\expFun{#2}{#3}}}}

\newcommand{\expFreeze}[1]{\texttt{freeze}\miniSepThree{#1}}

\newcommand{\expLet}[3]
  {{\twoThings{{\texttt{let}\miniSepThree\ensuremath{{#1}\miniSepThree{#2}}}}{#3}}}
\newcommand{\expLetLong}[3]
  {\ensuremath{\texttt{let}\miniSepThree{#1}\spaceThree{=}{#2}\spaceThree{\texttt{in}}{#3}}}
\newcommand{\expLetRec}[3]
  {\expThreeThings{\texttt{letrec}}{#1}{#2}{#3}}
\newcommand{\expIte}[3]
  {\texttt{if}\miniSepThree{#1}\spaceThree{#2}{#3}}
\newcommand{\expIteLong}[3]
  {\ensuremath{\texttt{if}\spaceThree{#1}\texttt{then}\spaceThree{#2}\texttt{else}\miniSepThree{#3}}}

\newcommand{\expCaseTwo}[3]
  {\expTwoThings{\texttt{case}}{#1}{\parens{#2}\miniSepThree{#3}}}

\newcommand{\expBranchWithId}[4]{{\expOneThing{#3}{#4}}}

\newcommand{\expEval}[1]{\texttt{eval}\miniSepThree{#1}}
\newcommand{\expApplyLens}[2]{\expTwoThings{\texttt{applyLens}}{#1}{#2}}

\newcommand{\regexextract}{\texttt{extract}}
\newcommand{\expRegexExtract}[2]{\expTwoThings{\regexextract}{#1}{#2}}
\newcommand{\regexreplace}{\texttt{replace}}
\newcommand{\expRegexReplace}[3]{\expThreeThings{\regexreplace}{#1}{#2}{#3}}

\newcommand{\expRecdExtend}[3]{\ensuremath{\ttlcurly{#1}\spaceThree{|}{#2}={#3}\ttrcurly}}
\newcommand{\expRecd}[2]{\ensuremath{\ttlcurly{#1}={#2}\ttrcurly}}
\newcommand{\expRecdDots}[2]{\ensuremath{\ttlcurly{#1}={#2};\ \cdots\ttrcurly}}
\newcommand{\expRecdTwo}[4]{\ensuremath{\ttlcurly{#1}={#2};\ {#3}={#4}\ttrcurly}}
\newcommand{\expRecdThree}[6]
  {\ensuremath{\ttlcurly{#1}={#2};\ {#3}={#4};\ {#5}={#6}\ttrcurly}}
\newcommand{\valEmptyRecd}{\ensuremath{\ttlcurly\ttrcurly}}

\newcommand{\expProjField}[2]{\ensuremath{{#1}.{#2}}}
\newcommand{\expProj}[2]{\expProjField{#1}{#2}}

\newcommand{\expProjStatic}[3]
  {\ensuremath{({#1}\miniSepTwo@\miniSepTwo{#2})[{#3}]}}

\newcommand{\ttfld}[1]{\ensuremath{\mathtt{{#1}}}}

\newcommand{\emptyEnv}{-}
\newcommand{\envCat}[2]{\ensuremath{{#1},\hspace{0.02in}{#2}}}
\newcommand{\envCatThree}[3]{\ensuremath{\envCat{#1}{\envCat{#2}{#3}}}}

\newcommand{\envBind}[2]{\ensuremath{{#1}\mapsto{#2}}}

\newcommand{\envExp}[2]
  {{#1}\vdash{#2}}

\newcommand{\evalArrowColor}
  {black}
\newcommand{\updateArrowColor}
  {black}

\newcommand{\reducesTo}[3]
  {\envExp{#1}{#2}\hspace{0.01in}\begingroup\color{\evalArrowColor}\Rightarrow\endgroup{#3}}

\newcommand{\simpleReducesTo}[2]
  {{#1}\Rightarrow{#2}}

\newcommand{\simpleUpdatesTo}[3]
  {{#1}\Leftarrow{#2}\rightsquigarrow{#3}}

\newcommand{\program}[2]
  {{#1}\vdash{#2}}
\newcommand{\updatesTo}[5]
  {\program{#1}{#2}\begingroup\color{\updateArrowColor}\Leftarrow\hspace{0.01in}\endgroup{#3}\rightsquigarrow\program{#4}{#5}}
\newcommand{\updatesToFancyUnderbrace}[6]
  {\program{#1}{#2}\Leftarrow\hspace{0.01in}{#3}\rightsquigarrow\underbrace{\program{#4}{#5}}_{#6}}
\newcommand{\pessimisticMark}{\checkmark}

\newcommand{\updatesToPessimistic}[5]
  {\program{#1}{#2}\begingroup\color{\updateArrowColor}
   \Leftarrow\endgroup{#3}\overset{\pessimisticMark}
   \rightsquigarrow\program{#4}{#5}}

\newcommand{\mergeEnvsAllArgs}[5]
  {{#2}\fourIdx{#4}{}{#5}{#1}\oplus{#3}}

\newcommand{\mergeEnvsTwoWay}[4]
  {\mergeEnvsAllArgs{}{#1}{#2}{#3}{#4}}
\newcommand{\mergeThree}[3]
  {\mergeEnvsAllArgs{#1}{#2}{#3}{}{}}
\newcommand{\combineEnvs}[3]{\mergeThree{#1}{#2}{#3}} %% TODO remove in favor of mergeEnvs

\newcommand{\mergeVals}[3]{\combineEnvs{#1}{#2}{#3}}
\newcommand{\mergeExps}[3]{\combineEnvs{#1}{#2}{#3}}

\newcommand{\equivEnvs}[3]{{#1}\equiv_{#3}{#2}}

\newcommand{\startAlignedPremises}{\begin{array}[b]{@{}r@{}c@{}l@{}}}
\newcommand{\stopAlignedPremises}{\end{array}}
\newcommand{\reducesToAligned}[3]
  {
   \envExp{#1}{#2} &
   \hspace{0.04in}\begingroup\color{\evalArrowColor}\Rightarrow\endgroup\hspace{0.035in} &
   {#3}
  }
\newcommand{\updatesToAligned}[5]
  {
   \program{#1}{#2} &
   \hspace{0.03in}\begingroup\color{\updateArrowColor}\Leftarrow\endgroup\hspace{0.045in} &
   {#3}\rightsquigarrow\program{#4}{#5}
  }
\newcommand{\updatesToAlignedPessimistic}[5]
  {
   \program{#1}{#2} &
   \hspace{0.03in}\begingroup\color{\updateArrowColor}\Leftarrow\endgroup\hspace{0.045in} &
   {#3}\overset{\pessimisticMark}{\rightsquigarrow}\program{#4}{#5}
  }
\newcommand{\equalsAligned}[2]
  {#1 &=& #2}

\newcommand{\benchmarks}{
\tableRow   {States Table A\hasvideo{}} { 37} { 304} {\plusminus{20}} {  2:36   } {  11 } {\nudgeRight57&\plusminus{5} & \nudgeRight154&\plusminus{20} & 85&\plusminus{20}&\speedup{200}} {  \fromto{1}  {2}     } {   1.18  } 
\tableRow   {States Table B\hasvideo{}} {126} { 774} {\plusminus{70}} {  0:43   } {  7  } {\nudgeRight256&\plusminus{40} & \nudgeRight456&\plusminus{50} & 331&\plusminus{80}&\speedup{700}} {  \always{1} } {1        } 
\tableRow   {Recipe\hasvideo{}   } {193} {1455} {\plusminus{80}} {  3:51   } {  17 } {\nudgeRight243&\plusminus{30} & \nudgeRight2237&\plusminus{200} & 1328&\plusminus{500}&\speedup{16}} {  \fromto{1}  {2}     } {   1.05  } 
\tableRow   {Budgetting          } { 37} { 328} {\plusminus{11}} {  0:45   } {  7  } {7&\plusminus{0.9} & \nudgeRight13&\plusminus{2} & 9&\plusminus{2}&\speedup{80}} {  \fromto{1}  {3}     } {    2    } 
\tableRow   {MVC                 } { 71} { 720} {\plusminus{50}} {  1:11   } {  10 } {\nudgeRight216&\plusminus{10} & \nudgeRight483&\plusminus{120} & 289&\plusminus{80}&\speedup{40}} {  \always{1} } {1        } 
\tableRow   {Linked-Text         } { 91} { 855} {\plusminus{40}} {  0:53   } {  5  } {\nudgeRight1886&\plusminus{140} & \nudgeRight2252&\plusminus{300} & 2025&\plusminus{200}&\speedup{5}} {  \fromto{1}  {2}     } {   1.2   } 
\tableRow   {Markdown            } {128} {1179} {\plusminus{110}} {  2:08   } {  6  } {\nudgeRight1369&\plusminus{90} & \nudgeRight1889&\plusminus{150} & 1607&\plusminus{200}&\speedup{13}} {  \always{1} } {1        } 
\tableRow   {Dixit               } {130} { 705} {\plusminus{40}} {  0:00   } {  15 } {\nudgeRight87&\plusminus{6} & \nudgeRight2205&\plusminus{4000} & 417&\plusminus{1500}&\speedup{120}} {  \always{1} } {1        } 
\tableRow   {Translation         } {122} { 357} {\plusminus{20}} {  0:00   } {  8  } {\nudgeRight187&\plusminus{12} & \nudgeRight1085&\plusminus{200} & 415&\plusminus{200}&\speedup{50}} {  \fromto{1}  {8}     } {    2    } 
\tableRow   {\LaTeX{} in HTML    } {534} {1648} {\plusminus{200}} {  0:00   } {  6  } {\nudgeRight413&\plusminus{50} & \nudgeRight3183&\plusminus{500} & 943&\plusminus{1000}&\speedup{150}} {  \always{1} } {1        } 
}
\newcommand{\benchmarksloc}{1469}
\newcommand{\benchmarkslocfloored}{1400}
\newcommand{\benchmarksnum}{10}
\newcommand{\benchmarksevalaverage}{833}
\newcommand{\benchmarksevalstddev}{\plusminus{400}}

\newcommand{\benchmarksnumupd}{92}
\newcommand{\benchmarksaverageoptupd}{723}

\newcommand{\benchmarksaveragestddev}{\plusminus{900}}
\newcommand{\benchmarksaveragespeedup}{\speedup{70}}

\newcommand{\benchmarkssolutionsaverage}{1.18}
\usepackage[firstpage]{draftwatermark}
\SetWatermarkText{\hspace*{4.2in}\raisebox{8.8in}{%
  \includegraphics[scale=0.33]{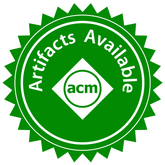}%
  \includegraphics[scale=0.33]{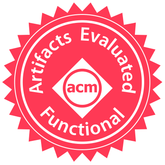}%
  \includegraphics[scale=0.33]{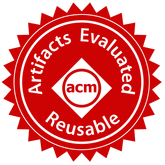}%
}}
\SetWatermarkAngle{0}

\begin{document}

\title
  {Bidirectional Evaluation with Direct Manipulation}

\author{Mika\"{e}l Mayer}
\affiliation{
  \institution{University of Chicago}            %% \institution is required
  \country{USA}                   %% \country is recommended
}
\email{mikaelm@uchicago.edu}          %% \email is recommended

\author{Viktor Kun\v{c}ak}
\affiliation{
  \institution{\'{E}cole Polytechnique F\'{e}d\'{e}rale de Lausanne}           %% \institution is required
  \country{Switzerland}                   %% \country is recommended
}
\email{viktor.kuncak@epfl.ch}         %% \email is recommended

\author{Ravi Chugh}
\affiliation{
  \institution{University of Chicago}
  \country{USA}                   %% \country is recommended
}
\email{rchugh@cs.uchicago.edu}

%% Abstract
%% \input{abstract-2}

\begin{abstract}

We present an \emph{evaluation update} (or simply, \emph{update}) algorithm
for a full-featured functional programming language,
which synthesizes program changes based on output changes.
Intuitively, the update algorithm retraces the steps of the original evaluation,
rewriting the program as needed to reconcile differences between the original
and updated output values.
Our approach, furthermore, allows expert users to define custom \emph{lenses}
that augment the update algorithm with more advanced or domain-specific program
updates.

To demonstrate the utility of evaluation update, we implement the algorithm in
\snsLeo{}, a novel \emph{direct manipulation programming system} for generating
HTML documents.
In \snsLeo{}, the user writes an ML-style functional program to generate HTML
output.
When the user directly manipulates the output using a graphical user
interface, the update algorithm reconciles the changes.
We evaluate bidirectional evaluation in \snsLeo{} by authoring
ten examples
comprising approximately \benchmarkslocfloored{}
lines of code in total.
These examples demonstrate how a variety of HTML documents and applications can
be developed and edited interactively in \snsLeo{}, mitigating the tedious
edit-run-view cycle in traditional programming environments.

\end{abstract}

%% 2012 ACM Computing Classification System (CSS) concepts
%% Generate at 'http://dl.acm.org/ccs/ccs.cfm'.
 \begin{CCSXML}
<ccs2012>
<concept>
<concept_id>10011007.10011006.10011008</concept_id>
<concept_desc>Software and its engineering~General programming
languages</concept_desc>
<concept_significance>500</concept_significance>
</concept>
<concept>
<concept_id>10011007.10011006.10011050.10011023</concept_id>
<concept_desc>Software and its engineering~Specialized application
languages</concept_desc>
<concept_significance>500</concept_significance>
</concept>
<concept>
<concept_id>10003120.10003121.10003124.10010865</concept_id>
<concept_desc>Human-centered computing~Graphical user interfaces</concept_desc>
<concept_significance>300</concept_significance>
</concept>
</ccs2012>
\end{CCSXML}

\ccsdesc{Software and its engineering~General programming languages}
\ccsdesc{Software and its engineering~Specialized application languages}
\ccsdesc{Human-centered computing~Graphical user interfaces}
%% End of generated code

%% Keywords
%% comma separated list
\keywords{
Bidirectional Programming,
Direct Manipulation,
Sketch-n-Sketch
}

\maketitle

\section{Introduction}
\label{sec:intro}

\newcommand{\limitationA}{Limitation A}
\newcommand{\limitationB}{Limitation B}
\newcommand{\limitationC}{Limitation C}
\newcommand{\limitationD}{Limitation D}
\newcommand{\limitations}{Limitations}

Expert programmers often choose to write programs to generate
digital objects that might otherwise be created in
graphical user interfaces (GUIs)
using \emph{direct manipulation}~\cite{Shneiderman1983,Hutchins:1985},
because GUIs typically lack powerful mechanisms for abstraction and reuse.
To name just a few example languages and libraries,
\LaTeX{} is particularly popular for generating documents;
JavaScript, Ruby, and Elm for \web{} applications;
Processing and p5.js (\url{http://p5js.org/}) for graphic designs;
\LaTeX{} and Slideshow~\cite{SlideshowJFP06} for
slide-based presentations; and
D3~\cite{BostockVIS2011} for data visualizations.

The benefits of programming, however, come at a steep cost: to
change the output of a program, the user must edit the source code, run it
again, and view the new output, often repeating this loop ad nauseaum.
Effort wasted in this way is particularly galling when
successive program changes---and the resulting output changes---are
small and narrow in scope.
Ideally, the user would ``directly manipulate'' the program output,
and the system would run the program ``in reverse'' to synthesize
necessary program repairs.

\parahead{Prior Approaches}

Two primary approaches help address the goal to run programs in reverse.
In \emph{bidirectional programming
languages}~\cite{lenses}, data transformations are
defined as \emph{lenses}, in which a $get$-function for forward-evaluation
is paired with a $put$-function for backward-evaluation---when the
output of $get$ is changed, $put$ specifies how to change the input.
A bidirectional language provides a set of domain-specific lens
primitives, from which programmers mold desired transformations using
\emph{lens combinators}.
Lenses have proven to be effective for defining bidirectional
transformations in a variety of domains---including structured data
(relational tables), semi-structured data (trees), unstructured data
(strings), and graphs.
Nevertheless, lenses are not a
solution for automatically reversing the computation of an arbitrary
program---data \emph{and} code---written in a general-purpose
functional language (\limitationA).

Another approach, developed by \citet{sns-pldi}, aims to reverse arbitrary
programs as follows.
First, the interpreter records \emph{value traces} to track the provenance of
how values are computed.
Then, when the user makes small changes to the output, updated value-trace
equations are solved in order to synthesize repairs to the program.
Although a useful step, this approach suffers several limitations.
First, the formulation supports tracing and updates for numeric values, but not
for other types of simple or more complex values (\limitationB).
Second, the formulation provides no way for expert users to customize the
behavior of the algorithm (\limitationC).
This is a significant limitation in practice, because no single update algorithm
for arbitrary programs can work well in all use cases.
Furthermore, even if extended to address the aforementioned limitations, the
approach requires that \emph{all} computations be traced even if many (or
most) values are not updated by the user.
For a large program where the subset of values that are
directly manipulated becomes a small fraction, the space overhead of this
approach could become a bottleneck, as is often the case for
systems that record execution traces, such as
\emph{omniscient debuggers}~\cite{Pothier2007}
and
\emph{query-based debuggers}~\cite{Ko:2008}
(\limitationD).

\parahead{Our Approach: Bidirectional Evaluation}

In contrast to prior approaches,
we propose a notion called \emph{bidirectional evaluation} for
programs in a full-featured, general-purpose functional programming language.
In addition to a standard evaluation relation
$\simpleReducesTo{\varExp}{\varVal}$ that evaluates expression $\varExp$ to
value $\varVal$, we define an \emph{evaluation update} (or simply,
\emph{update}) relation $\simpleUpdatesTo{\varExp}{\varVal'}{\varExp'}$ that,
given an expected value $\varVal'$, rewrites the original expression $\varExp$
to $\varExp'$.
Evaluation update proceeds by comparing the original output value $\varVal$ with
the goal $\varVal'$, and synthesizing repairs to $\varExp$ such that, ideally,
the new program $\varExp'$ evaluates to $\varVal'$.
Evaluation update is defined for arbitrary expressions $\varExp$ producing
arbitrary types of values $\varVal$, thus addressing \limitations{} A and B,
respectively.
Our approach relies on standard, uninstrumented evaluation; we
re-evaluate expressions as needed during update.
This approach trades time for space, thus addressing \limitationD.

Furthermore, we allow expert users to define custom lenses that augment the
update algorithm with more advanced or domain-specific program updates, thus
addressing \limitationC.
In particular, in place of an ordinary function application
$\expApp{\varExp_{\mathit{get}}}{\varExp}$, the program can define a
\emph{lens application}
$\expApplyLens
        {\expRecdTwo{\ttfld{apply}}{\varExp_{\mathit{get}}}
                    {\ttfld{update}}{\varExp_{\mathit{put}}}}
        {\varExp}$,
in which case, the update algorithm uses the designated \verb+update+
function $\varExp_{\mathit{put}}$ to help compute a new expression $\varExp'$
to replace the argument $\varExp$.

\parahead{Our Implementation: Direct Manipulation Programming for HTML}

We implement bidirectional evaluation within \snsLeo{}~\cite{sns-pldi}, an
interactive programming system for developing and editing graphical objects.
In the new system, the user writes a program in a functional,
ML-style language to generate HTML output.
When the user directly manipulates the output using a graphical user
interface, the update algorithm synthesizes repairs to reconcile the changes.
Our user interface provides a lightweight mechanism for previewing
and choosing a solution when there is ambiguity, inherent
to the setting of a general-purpose language.

We used our new version of \snsLeo{} to author \benchmarksnum{} examples comprising approximately
\benchmarkslocfloored{} lines in total, demonstrating how a variety of interactive
documents and applications---\web{} pages, Markdown-to-HTML
translators, scalable recipe editors,
and
what-you-see-is-what-you-get (WYSIWYG)
\LaTeX{} editors---can
be programmed in a way that allows direct manipulation changes to
propagate automatically back to the program.
Moreover, our prototype implementation typically
synthesizes program repairs for our examples in between 0 and 2 seconds,
which suggests that our techniques can be further developed and optimized for
more full-featured, interactive settings.

\parahead{Contributions and Outline}

To summarize, this paper provides the following contributions.

\begin{enumerate}

\item
We present the notion of \emph{bidirectional evaluation}, where arbitrary programs
in a general-purpose functional language can be run in reverse in order to
produce useful edits to the program.
To achieve this, we define an \emph{evaluation update} algorithm that---compared
to typical evaluation---receives an expected output value as an argument, used to
synthesize repairs to the expression such that it computes the expected value.
(\autoref{sec:leo-basic-update})

\item
We develop an approach for \emph{custom update lenses} that allow experts to
augment evaluation update with more advanced or domain-specific program updates.
To improve the utility of the ``built-in'' evaluation update algorithm, we show
how to define custom update lenses for several common functional programming
patterns.
(\autoref{sec:leo-custom-update})

\item
We implement our approach within \snsLeo{}; the new system is available on the \web{} at
\url{http://ravichugh.github.io/sketch-n-sketch/}.
Our implementation includes optimizations to make the update
algorithm perform well in practice, as well as programming conveniences
found in practical, ML-style functional languages.
Our examples and experiments demonstrate that the expressiveness and performance of
bidirectional evaluation in \snsLeo{} helps integrate the benefits of
programmatic and direct manipulation.
(\autoref{sec:implementation} and \autoref{sec:evaluation})

\end{enumerate}

\noindent
In the remainder of the paper, unqualified references to \snsLeo{} refer
to the new system.
Next, in \autoref{sec:overview}, we describe an overview example to introduce
the workflow enabled by bidirectional evaluation in \snsLeo{}, before describing
the approach in detail.

\section{Overview}
\label{sec:overview}

\setlength{\intextsep}{6pt}%
\setlength{\columnsep}{10pt}%

Consider the task for a \web{} developer to implement an HTML table that displays
each of the United States along with their capital cities.
In \snsLeo{}, the developer first writes a program in \leo{}---a functional
language that resembles Elm (\url{http://elm-lang.org/})---that generates a prototype.
The initial programming effort required to encode all intended data and
presentation constraints is similar to when using traditional text-based
programming environments.
Afterwards, however, \snsLeo{} allows the developer to:
(a) edit the data and design parameters through direct manipulation
interactions; and
(b) add rows to the table through a custom, library-defined user interface.
\snsLeo{} synthesizes program repairs based on these interactions.

\subsection{Initial Programming Effort}

%% Fig 0
%% \input{fig-overview-initial-program}

\begin{figure}[t]
\begin{center}
%% made on my laptop with in Chrome with 78% zoom...
\includegraphics[width=\linewidth]{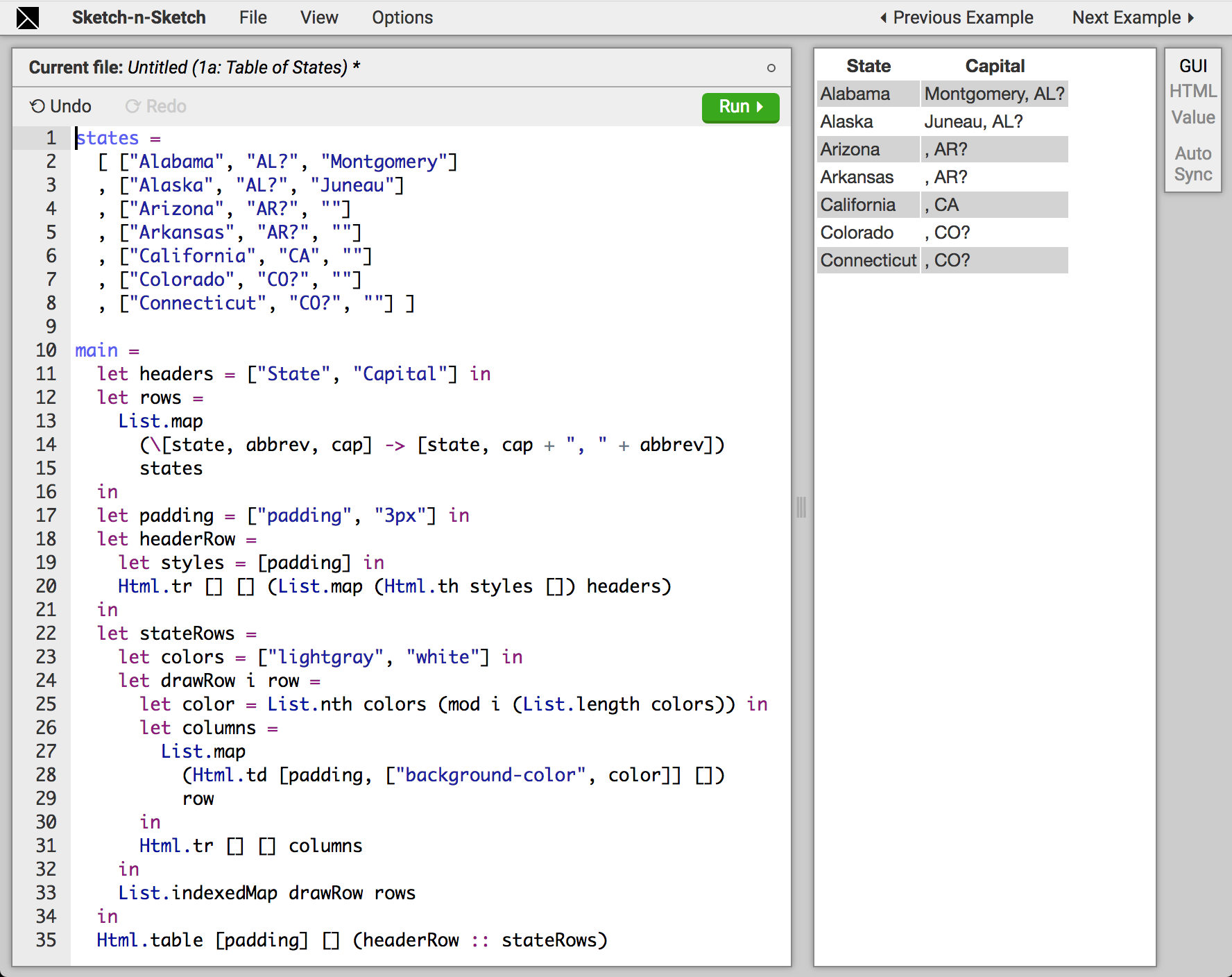}
\end{center}
\caption{A program written in \snsLeo{} that generates an HTML table of state
names and capital cities. This initial programming effort is performed with
traditional, text-based editing of source code.}
\label{fig:overview-initial-program}
\end{figure}

\autoref{fig:overview-initial-program} shows a \leo{} program
to generate an initial prototype.
In the following, we typeset string literals in the program with
typewriter font (\eg{}~\expStr{California}) and strings in the HTML output with
slanted font (\eg{}~\expOutputStr{California}).
Lines 1 through 8 define the table data; each element of
\verb+states+ is a three-element list, containing a state name, two-letter
abbreviation, and capital city.
For now, the data is incomplete---unknown abbreviations and capitals are
marked with question marks (\eg{}~\expStr{AL?} on lines 1--2) and
empty strings (\ie{}~\expStr{} on lines 4--8).

The \verb+main+ definition, starting on line 10, generates the output HTML
table.
First, the developer decides to produce two output columns: one for the state
name (\eg{}~\expOutputStr{Alabama}); and one for its capital city,
concatenated with the state abbreviation (\eg{}~\expOutputStr{Montgomery, AL}).
The \verb+headers+ definition (line 11) contains text for
the header row, and the \verb+rows+ definition (lines 12--15)
contains the text to display in subsequent rows by mapping
each three-element list
\verb+[state, abbrev, cap]+ in \verb+states+ to the two-element list
\verb-[state, cap + ", " + abbrev]-.
The \verb+headerRow+ definition (lines 18--20) uses library
functions \verb+Html.tr+ and \verb+Html.th+ to generate table row and
header elements, respectively, for the top of the table.
These \verb+Html+ functions take three arguments---a list of HTML
style attributes, a list of additional HTML attributes, and a list of HTML
child nodes---and produce straightforward encodings of HTML values to be
rendered.

The \verb+stateRows+ definition (lines 22--33) generates the
remaining rows of the table.
The \verb+colors+ list (line 23) defines two initial
colors---\expStr{lightgray} and \expStr{white}---and
the expression (line 25) chooses one of these colors based on the parity
of row index \verb+i+ (received as a parameter from the
\verb+List.indexedMap+ library function).
The \verb+columns+ definition (lines 26--29) places the text
for each state and its capital city---in a two-element list
\verb+row+---inside \verb+Html.td+ elements, which comprise a row
built from the \verb+Html.tr+ expression (line 31).

Lastly, the expression on line 35 builds the overall \verb+Html.table+
element comprising \verb+headerRows+ and \verb+stateRows+.
The output \leo{} value is translated to HTML in a straightforward manner
and rendered graphically in the
right half of \snsLeo{}, as shown in
\autoref{fig:overview-initial-program}.

\subsection{Direct Manipulation of Output Text}
\label{sec:overview-text-edit}

%% Fig 1
%% \input{fig-overview-text-edit-AL-AK}
\begin{figure}[t]
\begin{center}
\includegraphics[width=\linewidth]{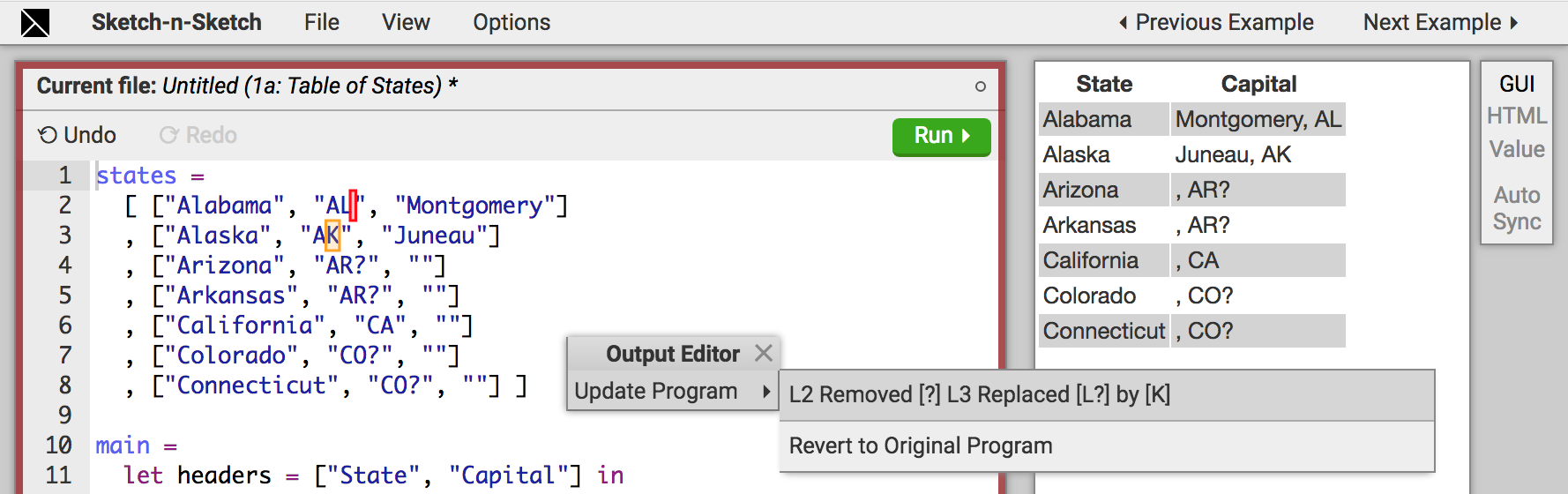}
\caption{
Direct Manipulation Text Edit.
(1) In the \expOutputStr{Capital} column, the user changes \expOutputStr{AL?} in
the first row to \expOutputStr{AL} and \expOutputStr{AL?} in the second row to
\expOutputStr{AK}.
(2) The output is out of synchronization with the program, so the
editor displays a pop-up menu.
(3) The user hovers the \valuePropToolName{} menu item, which then displays
a nested menu with one possible program repair.
(4) The user hovers this option, which previews the new code and output in the
left and right panes, respectively.
The screenshot captures this last step.
}
\label{fig:overview-text-edit-AL-AK}
\end{center}
\end{figure}

Having encoded the intended programmatic relationships for the data and
design of the table, next the developer wants to correct the missing
data (lines 2--8).
In \snsLeo{}, the developer can edit text directly in the graphical user
interface that displays the output (the right half of the editor).
The interactions described in this subsection, and the following ones, can be
viewed in screencast videos, available on the \web{}.

\newcommand{\paraComputingDisplaying}{Computing and Displaying Program Updates}
\parahead{\paraComputingDisplaying}

\autoref{fig:overview-text-edit-AL-AK} shows an example of how the developer
edits the data in the program through the graphical user interface interface;
the screenshot shows the editor state \emph{after} the following sequence of
user actions.

First, in the first state row of the output table, the user deletes the question
mark after \expOutputStr{AL} in the string \expOutputStr{, AL?}.
Next, in the second row, the user replaces the string \expOutputStr{AL?} with
\expOutputStr{AK}.
As soon as the user begins editing the output table, \snsLeo{} detects that the
program output is no longer synchronized with the program.
As a result, \snsLeo{} highlights the code box with a red border and
displays a pop-up window with a menu item labeled \valuePropToolName{}.

When the user hovers over \valuePropToolName{}, \snsLeo{} runs the evaluation
update algorithm to synthesize a repaired program that, when re-evaluated,
generates the same result as the directly manipulated output.
In this case, the algorithm computes one solution that, along with an option for
reverting the changes, is displayed in a nested menu to the
right of \valuePropToolName{}.
The screenshot in \autoref{fig:overview-text-edit-AL-AK} captures the editor
state when the user hovers over the first item in the nested menu, at which
point \snsLeo{} displays a preview of the updated program (resp. output)
directly in the left (resp. right) pane.
The caption \expOutputStr{L2 Removed [?] L3 Replaced [L?] by [K]} summarizes the
\emph{string differences}, on lines 2 and 3, between the original and updated
program text.
These differences are highlighted in red and orange in the code box to further
help communicate the proposed changes to the user.
In this case, the new program matches the user's expectations, so the user
clicks the menu item (not shown in the screenshot) to confirm the update,
returning the program and output to a synchronized state.

\newcommand{\paraAmbiguity}{Ambiguity}
\parahead{\paraAmbiguity}

%% Fig 2
%% \input{fig-overview-text-edit-AZ}
\begin{figure}[t]
\begin{center}
\includegraphics[width=\linewidth]{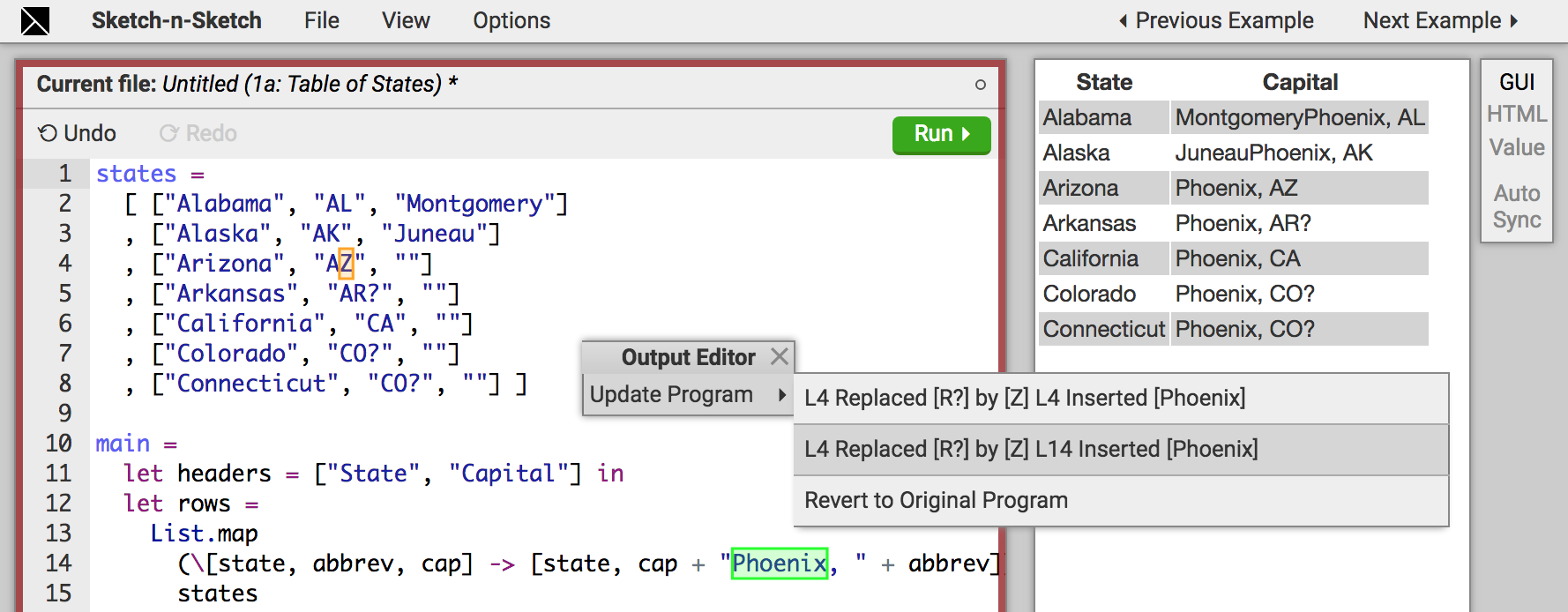}
\caption{
Direct Manipulation Text Edit with Ambiguity.
(1) In the \expOutputStr{Capital} column, the user changes \expOutputStr{?, AR?} in
the third row to \expOutputStr{Phoenix, AZ}.
(2) The editor displays a pop-up menu.
(3) The user hovers the \valuePropToolName{} menu item, resulting in two
candidate program repairs.
(4) The user hovers over the second one, which previews the (undesirable)
changes to the code and output.
The screenshot captures this last step.
}
\label{fig:overview-text-edit-AZ}
\end{center}
\end{figure}

Whereas each of the previous output changes resulted in a single solution,
\autoref{fig:overview-text-edit-AZ} shows a change that leads to multiple.
In the third row, the user replaces \expOutputStr{, AR?} with
\expOutputStr{Phoenix, AZ}.
When the \valuePropToolName{} menu item appears and is hovered, two solutions
are displayed (in addition to the option to revert the changes).
The screenshot in \autoref{fig:overview-text-edit-AZ} captures the editor state
when the second solution is hovered.
Both solutions replace \expStr{AR?} on line 4 with \expStr{AZ}, as desired, but
the second solution inserts \expStr{Phoenix} as a prefix to the \expStr{, }
separator string used in the concatenation on line 14.
By viewing the preview of the output---with \expOutputStr{Phoenix} appearing
in all rows---the user quickly determines that this change, though consistent
with the output edit, is undesirable.
So, the user hovers and selects the first option (not shown in the screenshot).
Wanting the separator string on line 14 always to remain constant,
the developer edits the source code to wrap the string \expStr{", "} in a
call to \verb+Update.freeze+ (not shown), which instructs \snsLeo{}
never to change this expression when computing program updates.

\parahead{Browser Conveniences for Navigating Output Text}

The user fills in missing data for the remaining rows directly in the output
pane.
Having frozen the separator string already, none of these changes lead to
ambiguity.
During these interactions, the user benefits from text-editing features built-in
to the browser---using the \key{Tab} key to advance to subsequent columns and rows,
and arrow keys to navigate the text cursor within the selected cell---which make
it yet more convenient to specify these changes in the graphical user interface
rather than in the source code editor.

\subsection{Direct Manipulation with DOM Inspector}
\label{sec:overview-inspector-edit}

%% Fig 3
%% \input{fig-overview-inspector-edit}
\begin{figure}[t]
\begin{center}
\includegraphics[width=\linewidth]{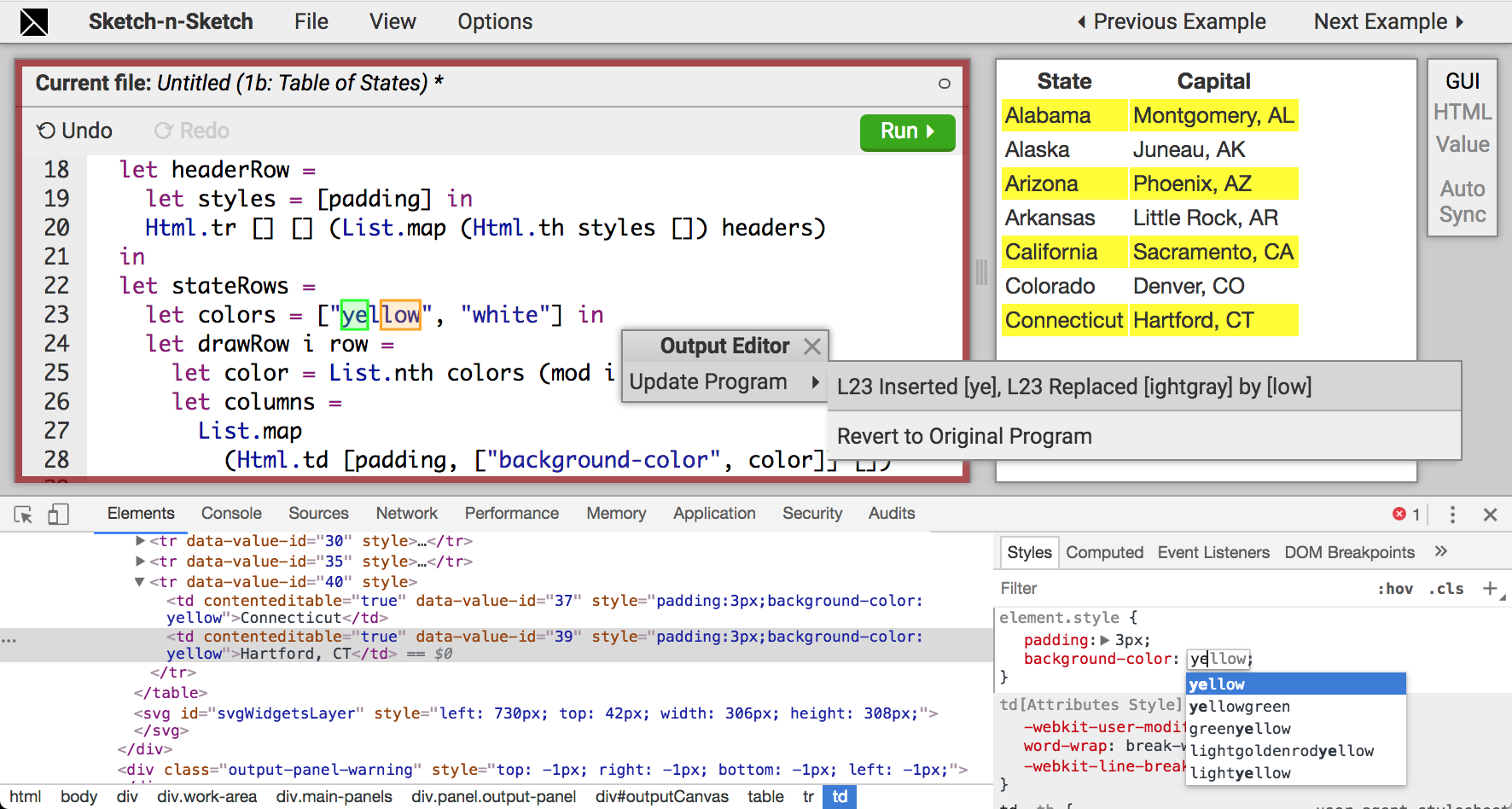}
\end{center}
\caption{
Direct Manipulation DOM Inspector Edit.
(1) The user right-clicks the \expOutputStr{Hartford, CT} cell and
selects \toolName{Inspect} from a built-in browser pop-up menu.
(2) The \browserToolName{DOM Inspector} pane appears at the bottom of the editor window with the
selected table data element in focus.
(3) The user changes the \texttt{background-color} of the selected table
cell from \texttt{lightgray} to \texttt{yellow}.
(4) The editor displays a pop-up menu.
(5) The user hovers the \valuePropToolName{} menu item, resulting in one
program repair.
(6) The user hovers over the option, previewing the effect of the new color on
all alternating rows.
The screenshot captures this last step.
}
\label{fig:overview-inspector-edit}
\end{figure}

Having corrected the table data
(as shown in \autoref{fig:overview-inspector-edit}),
next the developer experiments with different styles.
\snsLeo{} allows the developer to use the existing Developer Tools
provided by modern browsers to \emph{inspect} and \emph{modify} arbitrary
elements and attributes in the DOM (\ie{}~the HTML output of the program);
these changes are used to trigger the update algorithm.

\parahead{Browser Conveniences for Editing Styles}

Suppose the user wants to try out different colors for alternating rows, to
replace the colors on line 23.
\autoref{fig:overview-inspector-edit} shows how the user can affect such
changes.
First, the user right-clicks the \expOutputStr{Hartford, CT} cell and selects
\browserToolName{Inspect} from the browser's pop-up menu.
As a result, a \browserToolName{Developer Tools} pane appears at the bottom of
editor (as shown in the screenshot), with the selected cell in focus in the
\browserToolName{DOM Element Inspector}.
The rightmost panel provides a \browserToolName{Styles Editor},
which the developer can use to
change the \verb+background-color+ from the initial \verb+lightgray+ color.
The developer starts typing \verb+ye+ and, then, using
the built-in conveniences provided by the \browserToolName{Styles Editor} for
changing color
values---a dropdown menu of related colors, equipped with tab completion and
previews---decides to try the color \verb+yellow+.

As with the text changes before, \snsLeo{} detects that the output is no longer
synchronized with the program, and so displays the \valuePropToolName{} menu.
The screenshot in \autoref{fig:overview-inspector-edit} captures the editor
state when the user hovers the single solution, which replaces \expStr{lightgray} on
line 23 with \expStr{yellow} to reconcile the change.
In the result of this new program, the color of all cells in alternating
rows are changed, not just the one cell directly manipulated.

\newcommand{\paraAutoSync}{Automatic Synchronization}
\parahead{\paraAutoSync}

The developer wants to experiment with colors,
but manually hovering and clicking the \valuePropToolName{} menu will
be tedious when trying several options.
So, the developer clicks the button labeled
\toolName{Auto Sync} in the right toolbar.
Now whenever the output is changed,
the program update algorithm is
automatically run after a delay---1000ms by default, but this choice
can be configured by the program.
When there is a single solution, it is applied automatically.
Thus, the developer can try several colors in the
\browserToolName{DOM Inspector} in rapid fashion, viewing how the change propagates (almost)
immediately to the entire table.

\parahead{Small Updates}

Next, the developer wishes to add a background color to \verb+headerRow+, whose
\verb+styles+ list on line 19 does not include a color.
The following interactions are not depicted in screenshots, but they
are in the accompanying videos.
After selecting the first column of the header row in the browser
\browserToolName{DOM Inspector} (either by right-clicking, or using the browser's
built-in \browserToolName{Inspect} cursor), the
developer, again, uses the \browserToolName{Styles Editor},
which provides an easy way (with a mouse
click or \key{Enter} key press) to add a new attribute.
The user adds
a new \verb+background-color+ attribute set to the value \verb+orange+,
and the corresponding program update adds the pair
$\expListTwo{\expStr{background-color}}{\expStr{orange}}$ to the \verb+styles+
list on line 19.
Unlike all of the \emph{local updates}, described above, which replaced only
constant literals in the program with new ones, this solution includes
a \emph{structural update} that alters the structure of the abstract syntax tree.
We call this structural update \emph{pretty local} because the only change to the
structure is inserting a new literal at a leaf of the AST, \ie{}~inside another
list literal.
In our experience (\autoref{sec:evaluation}), even just ``small'' program
updates---local and pretty local---enable a variety of desirable interactions.

\subsection{Direct Manipulation with Custom User Interfaces}
\label{sec:overview-add-row}

%% Fig 4
%% \input{fig-overview-add-row}
\begin{figure}[t]
\begin{center}
\includegraphics[width=\linewidth]{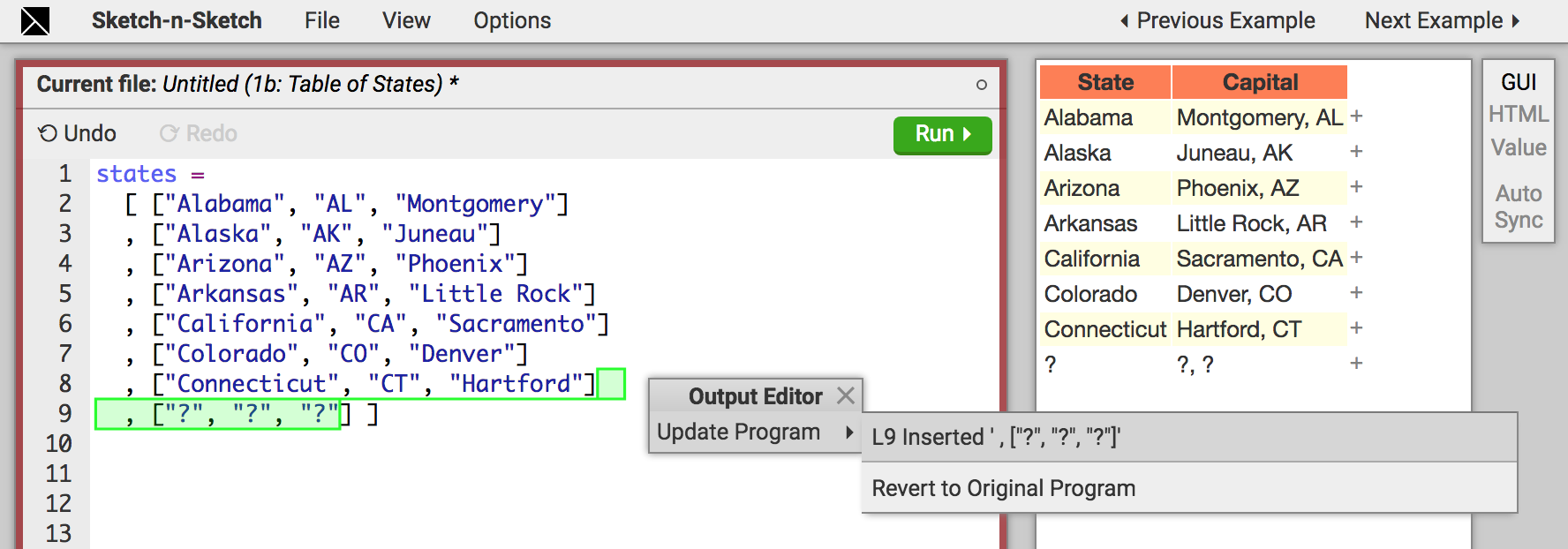}
\end{center}
\caption{Custom Add Row Button.
(0) Initially, the last table row is \expOutputStr{Connecticut}.
(1) The user clicks the \expOutputStr{+} button to the right of the
\expOutputStr{Connecticut} row.
(2) The editor displays a pop-up menu.
(3) The user hovers \valuePropToolName{}, resulting in one
program repair.
(4) The user hovers over the option, which previews a new program that
generates an additional placeholder row after \expOutputStr{Connecticut}.
The screenshot captures this last step.
}
\label{fig:overview-add-row}
\end{figure}

What about adding a new row with columns \expOutputStr{Delaware} and
\expOutputStr{Dover, DE} to the bottom of the output table?
The desired program repair is to add a new three-element list
$\expListThree{\expStr{Delaware}}{\expStr{DE}}{\expStr{Dover}}$.
to the end of the \verb+states+ list.
As we will explain in \autoref{sec:leo}, our algorithm does not
automatically perform the complex reasoning required---pushing
the newly inserted value \expListTwo{\expStr{DE}}{\expStr{Dover, DE}}
back through the call to \verb+List.map+ on line 13 of
\autoref{fig:overview-text-edit-AZ}---to produce the desired program repair.

\parahead{User-Defined Program Updates with Lenses}

For situations like this,
\snsLeo{} provides users (or library writers) a
mechanism to define a custom \emph{lens} that augments a
``bare'' function with a second \verb+update+ function that defines the
``reverse semantics'' for the bare function.

Using lenses, the expert user can define a module called
\verb+TableWithButtons+---which performs more advanced evaluation update
than for basic \verb|List.map|---to serve as a
drop-in replacement for the basic table-constructing functions in the
\verb+Html+ library.
\autoref{fig:overview-add-row} shows how, using this library, the user can click a
button (labeled \expOutputStr{+}) to indicate that a new row should be added
at that position.
Suppose the user clicks the button next to the \expOutputStr{Connecticut} row, hovers
\valuePropToolName{}, and then hovers the single solution (as shown in the
screenshot).
The resulting program
adds a dummy row on line 9 in the
\verb+states+ list, which can later be filled in through the basic direct
manipulation text interactions as before.
Thus, by using lenses to augment the functionality of the built-in
update algorithm, expert users and library writers can implement custom user
interface features for manipulating the particular bidirectional functional
documents under construction.

\section{\littleLeo{}: A Bidirectional Functional Programming Language}
\label{sec:leo}

In this section, we present the syntax and bidirectional evaluation semantics
for \littleLeo{}, a lambda-calculus that models the
\leo{} language supported in our implementation.\footnote{
\littleLeo{} is a backronym for ``a \underline{Little}
language extended with \underline{L}enses, \underline{E}val,
\underline{O}bjects, and other features.''}

\parahead{Syntax}

\begin{figure}[t]

\small

\centering

\figSyntaxBegin

\figSyntaxRowLabel{Expressions}{\varExpBare}
  \varConst
  \figSyntaxSpaceItem
  \varVar
  \figSyntaxSpaceItem
  \expFun{\varPat}{\varExp}
  \figSyntaxSpaceItem
  \expApp{\varExp_1}{\varExp_2}
  \figSyntaxSpaceItem
  \expCons{\varExp_1}{\varExp_2}
  \figSyntaxSpaceItem
  \expRecdExtend{\varExp}{\varField}{\varExp_{\varField}}
  \figSyntaxSpaceItem
  \expProjField{\varExp}{\varField}
\figSyntaxLineBreak
\figSyntaxRow
  \expLet{\varPat}{\varExp_1}{\varExp_2}
  \figSyntaxSpaceItem
  \expLetRec{\varPat}{\varExp_1}{\varExp_2}
  \figSyntaxSpaceItem
  \expIte{\varExp_1}{\varExp_2}{\varExp_3}
  \figSyntaxSpaceItem
  \expCaseTwo{\varExp}{\expBranchWithId{\varExpId}{1}{\varPat_1}{\varExp_1}}{\cdots}
\figSyntaxLineBreak
\figSyntaxRow
  \expFreeze{\varExp}
  \figSyntaxSpaceItem
  \expApplyLens{\varExp_1}{\varExp_2}
\figSyntaxSpaceNextCategory
\figSyntaxRowLabel{Constants}{\varConst}
  \varNum
  \figSyntaxSpaceItem
  \varBool
  \figSyntaxSpaceItem
  \varStr
  \figSyntaxSpaceItem
  \valNil
  \figSyntaxSpaceItem
  \valEmptyRecd
  \figSyntaxSpaceItem
  \op{(+)}
  \figSyntaxSpaceItem
  \op{(*)}
  \figSyntaxSpaceItem
  \op{(++)}
  \figSyntaxSpaceItem
  \op{(\&\&)}
  \figSyntaxSpaceItem
  \op{not}
  \figSyntaxSpaceItem
  \cdots
\figSyntaxLineBreak
\figSyntaxRow
  \op{updateApp}
  \figSyntaxSpaceItem
  \op{diff}
  \figSyntaxSpaceItem
  \op{merge}
\figSyntaxSpaceNextCategory
\figSyntaxRowLabel{Patterns}{\varPatBare}
  \varConst
  \figSyntaxSpaceItem
  \varVar
  \figSyntaxSpaceItem
  \expCons{\varPat_1}{\varPat_2}
  \figSyntaxSpaceItem
  \expRecdDots{\varField_1}{\varPat_1}

\figSyntaxSpaceNextCategory
\\

\figSyntaxRowLabel{Environments}{\varEnv}
  \emptyEnv
  \figSyntaxSpaceItem
  \envCat{\varEnv}{\envBind{\varPat}{\varVal}}
\figSyntaxSpaceNextCategory
\figSyntaxRowLabel{Values}{\varVal} %% ,\varValW}
  \varConst
  \figSyntaxSpaceItem
  \closure{\varEnv}{\varPat}{\varExp}
  \figSyntaxSpaceItem
  \expListTwo{\varVal_1}{\cdots}
  \figSyntaxSpaceItem
  \expRecdDots{\varField_1}{\varVal_1}
\figSyntaxEnd
\caption{Syntax of \littleLeo{}.}
\label{fig:syntax}
\end{figure}

The syntax of \littleLeo{} is defined in \autoref{fig:syntax}.
Constants include numbers $\varNum$, booleans $\varBool$, strings $\varStr$, the
empty list $\valNil$, the empty record $\valEmptyRecd$, and primitive operators.
The operators \op{updateApp}, \op{diff}, and \op{merge}
facilitate the definition of custom lenses, which will be presented after
basic evaluation update.
Values $\varVal$ include constants, closures
$\closure{\varEnv}{\varPat}{\varExp}$ where the environment $\varEnv$ binds free
variables in the body of the function $\expFun{\varPat}{\varExp}$, and lists and
records with multiple components.

The definition of expression forms is spread across three lines:
those on the first two lines are standard---constants $\varConst$,
variables $\varVar$,
function application $\expApp{\varExp_1}{\varExp_2}$,
list construction $\expCons{\varExp_1}{\varExp_2}$,
record extension $\expRecdExtend{\varExp}{\varField}{\varExp_{\varField}}$,
record field projection $\expProj{\varExp}{\varField}$,
(simple and recursive) let-bindings $\expLet{\varVar}{\varExp_1}{\varExp_2}$,
conditionals $\expIte{\varExp_1}{\varExp_2}{\varExp_3}$, and
case expressions
$\expCaseTwo{\varExp}{\expBranchWithId{\varExpId}{1}{\varPat_1}{\varExp_1}}{\cdots}$;
and those on the third line are specific to the
definition of custom update functions
(discussed in \autoref{sec:leo-custom-update}).

\subsection{Bidirectional Evaluation}
\label{sec:leo-basic-update}

\begin{figure}[t]

\small

\centering

\judgementHeadCenter
  {Evaluation\phantom{p}}
  {$\reducesTo{(\varEnv}{\varExp)}{\varVal}$}
\hspace{0.35in}
\judgementHeadCenter
  {Evaluation Update\phantom{p}}
  {$\updatesTo{(\varEnv}{\varExp)}{\varVal'}{(\varEnv'}{\varExp')}$}

\vsepRule

$\inferrule*[right=\ruleNameFig{E-Const}]
  {
  }
  {\reducesTo{\varEnv}{\varConst}{\varConst}}
$
\hsepRule
$\inferrule*[right=\ruleNameFig{U-Const}]
  {
  }
  {\updatesTo{\varEnv}{\varConst}{\varConst'}{\varEnv}{\varConst'}}
$

\vsepRule

$\inferrule*[right=\ruleNameFig{E-Fun}]
  {
  }
  {\reducesTo
    {\varEnv}
    {\expFun{\varPat}{\varExp}}
    {\closure{\varEnv}{\varPat}{\varExp}}
  }
$
\hsepRule
$\inferrule*[right=\ruleNameFig{U-Fun}]
  {
  }
  {\updatesTo
    {\varEnv}
    {\expFun{\varPat}{\varExp}}
    {\closure{\varEnv'}{\varPat}{\varExp'}}
    {\varEnv'}
    {\expFun{\varPat}{\varExp'}}
  }
$

\vsepRule

$\inferrule*[right=\ruleNameFig{E-Var}]
  {
   \varEnv=\envCatThree{\varEnv_1}{\envBind{\varVar}{\varVal}}{\varEnv_2}
   \sepPremise
   \varVar\notin\dom{\varEnv_2}
  }
  {\reducesTo{\varEnv}{\varVar}{\varVal}}
$
\hsepRule
$\inferrule*[right=\ruleNameFig{U-Var}]
  {
   \varEnv=\envCatThree{\varEnv_1}{\envBind{\varVar}{\varVal}}{\varEnv_2}
   \sepPremise
   \varVar\notin\dom{\varEnv_2}
  }
  {\updatesTo{\varEnv}{\varVar}{\varVal'}
    {(\envCatThree{\varEnv_1}{\envBind{\varVar}{\varVal'}}{\varEnv_2})}
    {\varVar}}
$

\vsepRule

$\inferrule*[right=\ruleNameFig{E-Let}]
  {
  {
  \startAlignedPremises{}
   \reducesToAligned
     {\varEnv}
     {\varExp_1}
     {\varVal_1}
   \\
   \reducesToAligned
    {\envCat{\varEnv}{\envBind{\varVar}{\varVal_1}}}
    {\varExp_2}
    {\varVal_2}
  \stopAlignedPremises{}
  }
  }
  {\reducesTo
    {\varEnv}
    {\expLet{\varVar}{\varExp_1}{\varExp_2}}
    {\varVal_2}
  }
$
\hspace{0.12in} %% for arxiv-generated PDF
%% \hsepRule
%% \hfill
%
$\inferrule*[right=\ruleNameFig{U-Let}]
  {
  {
  \startAlignedPremises{}
   \reducesToAligned
     {\varEnv}
     {\varExp_1}
     {\varVal_1}
   \\
   \updatesToAligned
    {\envCat{\varEnv}{\envBind{\varVar}{\varVal_1}}}
    {\varExp_2}
    {\varVal_2'}
    {\envCat{\varEnv_2}{\envBind{\varVar}{\varVal_1'}}}
    {\varExp_2'}
  \stopAlignedPremises{}
  }
  \hspace{0.13in}
  {
  \startAlignedPremises{}
   \updatesToAligned
    {\varEnv}
    {\varExp_1}
    {\varVal_1'}
    {\varEnv_1}
    {\varExp_1'}
   %% \\
   %% \phantom{\ }
   \\
   \equalsAligned{\varEnv'}{\mergeEnvsAllArgs{\varEnv}{\varEnv_1}{\varEnv_2}{\varExp_1}{\varExp_2}}
  \stopAlignedPremises{}
  }
  }
  {\updatesTo
    {\varEnv}
    {\expLet{\varVar}{\varExp_1}{\varExp_2}}
    {\varVal_2'}
    {\varEnv'}
    %% {\combineEnvs{\varEnv}{\varEnv_1}{\varEnv_2}}
    {\expLet{\varVar}{\varExp_1'}{\varExp_2'}}
  }
$

\vsepRule

$\inferrule*[right=\ruleNameFig{E-App}]
  {
  %% {
  %% \startAlignedPremises{}
   \reducesTo{\varEnv}{\varExp_1}{\closure{\varEnv_f}{\varVar}{\varExp_f}}
   \sepPremise
   \reducesTo{\varEnv}{\varExp_2}{\varVal_2}
   \sepPremise
   \reducesTo{\envCat{\varEnv_f}{\envBind{\varVar}{\varVal_2}}}{\varExp_f}{\varVal}
  %% \stopAlignedPremises{}
  %% }
  }
  {\reducesTo
    {\varEnv}
    {\expApp{\varExp_1}{\varExp_2}}
    {\varVal}
  }
$

\vsepRule

$\inferrule*[right=\ruleNameFig{U-App}]
  {
  {
  \startAlignedPremises{}
   \reducesToAligned{\varEnv}{\varExp_1}{\closure{\varEnv_f}{\varVar}{\varExp_f}}
   \\
   \reducesToAligned{\varEnv}{\varExp_2}{\varVal_2}
   \\
   \updatesToAligned
     {\envCat{\varEnv_f}{\envBind{\varVar}{\varVal_2}}}
     {\varExp_f}
     {\varVal'}
     {\envCat{\varEnv_f'}{\envBind{\varVar}{\varVal_2'}}}
     {\varExp_f'}
  \stopAlignedPremises{}
  }
  \\
  {
  \startAlignedPremises{}
   \updatesToAligned
     {\varEnv}
     {\varExp_1}
     {\closure{\varEnv_f'}{\varVar}{\varExp_f'}}
     {\varEnv_1}
     {\varExp_1'}
   \\
   \updatesToAligned
     {\varEnv}
     {\varExp_2}
     {\varVal_2'}
     {\varEnv_2}
     {\varExp_2'}
   %% \\
   %% \phantom{\ }
   \\
   %% \equalsAligned{\varEnv'}{\combineEnvs{\varEnv}{\varEnv_1}{\varEnv_2}}
   \equalsAligned{\varEnv'}{\mergeEnvsAllArgs{\varEnv}{\varEnv_1}{\varEnv_2}{\varExp_1}{\varExp_2}}
  \stopAlignedPremises{}
  }
  }
  {\updatesTo
    {\varEnv}
    {\expApp{\varExp_1}{\varExp_2}}
    {\varVal'}
    {\varEnv'}
    %% {\combineEnvs{\varEnv}{\varEnv_1}{\varEnv_2}}
    {\expApp{\varExp_1'}{\varExp_2'}}
  }
$

\vsepRule

$\inferrule*[right=\ruleNameFig{E-If-True}]
  {
  {
  \startAlignedPremises{}
   \reducesToAligned
     {\varEnv}
     {\varExp_1}
     {\expTrue}
   \\
   \reducesToAligned
     {\varEnv}
     {\varExp_2}
     {\varVal}
  \stopAlignedPremises{}
  }
  }
  {\reducesTo
    {\varEnv}
    {\expIte{\varExp_1}{\varExp_2}{\varExp_3}}
    {\varVal}
  }
$
\hsepRule
$\inferrule*[right=\ruleNameFig{U-If-True}]
  {
  {
  \startAlignedPremises{}
     \reducesToAligned{\varEnv}{\varExp_1}{\expTrue} \\
     \updatesToAligned{\varEnv}{\varExp_2}{\varVal'}{\varEnv_2}{\varExp_2'}
  \stopAlignedPremises{}
  }
  \\
  \varEnv' = \mergeEnvsAllArgs{\varEnv}{\varEnv}{\varEnv_2}{\varExp_1}{\varExp_2}
  }
  {\updatesTo
    {\varEnv}
    {\expIte{\varExp_1}{\varExp_2}{\varExp_3}}
    {\varVal'}
    {\varEnv'}
    {\expIte{\varExp_1}{\varExp_2'}{\varExp_3}}
  }
$

\vsepRule

$\inferrule*[right=\ruleNameFig{E-Freeze}]
  {\reducesTo{\varEnv}{\varExp}{\varVal}}
  {\reducesTo{\varEnv}{\expFreeze{\varExp}}{\varVal}}
$
\hsepRule
$\inferrule*[right=\ruleNameFig{U-Freeze}]
  {
   \reducesTo{\varEnv}{\varExp}{\varVal}
  }
  {\updatesTo{\varEnv}{\expFreeze{\varExp}}{\varVal}{\varEnv}{\expFreeze{\varExp}}}
$

\caption{Bidirectional Evaluation for \littleLeo{} (selected rules).}
\label{fig:eval-update}
\end{figure}

%% \autoref{fig:eval-update} defines the \emph{bidirectional evaluation} semantics
\autoref{fig:eval-update}, \autoref{fig:update-numbers}, and
\autoref{fig:update-cons} define the \emph{bidirectional evaluation} semantics
for \littleLeo{}; the big-step evaluation rules (prefixed \ruleName{``E-''}) are standard,
while the \emph{evaluation update} (or simply, \emph{update}) rules (prefixed
\ruleName{``U-''}) are novel.\footnote{
%%
%% for several
%% \littleLeo{} expression forms, with standard big-step evaluation rules on the
%% left and novel \emph{evaluation update} (or simply, \emph{update}) rules on
%% the right.\footnote{
By analogy to \emph{bidirectional
type checking}~\cite{PierceTurner,Chlipala:2005da},
evaluation may be thought of as ``value synthesis'' and evaluation update as
``value checking.''}
We refer to an environment-expression pair $\program{\varEnv}{\varExp}$ as a
\emph{program}.
%
%% The judgement $\updatesTo{\varEnv}{\varExp}{\varVal'}{\varEnv'}{\varExp'}$ %
%states that ``when updating its output value to $\varVal'$, the expression %
%$\varExp$ (in environment $\varEnv$) updates to $\varExp'$ (with new
%environment % $\varEnv'$).''

The evaluation update judgement
$\updatesTo{(\varEnv}{\varExp)}{\varVal'}{(\varEnv'}{\varExp')}$
states that ``when updating its output value to $\varVal'$, the program
$\program{\varEnv}{\varExp}$ updates to $\program{\varEnv'}{\varExp'}$.''
%
%% Parentheses in the judgement are for clarity; we typically omit them.
The parentheses are for clarity; we typically omit them.
When only the expression (resp. environment) changes, we sometimes say ``the
expression (resp. environment) updates to a new expression (resp.
environment).''
We sometimes say ``push $\varVal'$ (or changes to $\varVal$) back to $\varExp$''
in reference to evaluation update.
%
%% in reference to a judgement of the form %
%$\updatesTo{\varEnv}{\varExp}{\varVal'}{\varEnv'}{\varExp'}$.
%
%% The update judgement does not refer to the original value $\varVal$
%% produced by the program; if needed by a premise of an update rule, it is
%% re-computed.
The judgement does not refer to the original value $\varVal$
from forward-evaluation; if needed by a premise of an update rule, it is
re-computed.
Update rules come in three varieties:
\emph{replacement rules} overwrite values (base constants and closures)
in the program with new ones;
%
%% \emph{primitive rules} define how to invert calls to primitive operators; and
\emph{primitive rules} define how to update operations on base values, lists,
and dictionaries; and
\emph{propagation rules} carry the effects of replacement and primitive
rules to the rest of the program (through variables, function calls,
conditionals, \etc{}).

%% \parahead{Simple Rules}

\subsubsection{Replacement Rules}

%% \parahead{Values and Variables}

%% We start with a few update rules that do not recursively refer to the update
%% judgement.

There are two axioms, for replacing values.
The rule \ruleName{U-Const} says that, when updating its output value to
$\varConst'$, the expression $\varConst$ updates to $\varConst'$; the
environment $\varEnv$ remains unchanged.
The rule \ruleName{U-Fun} says that, when updating its output value to the
closure $\closure{\varEnv'}{\varPat}{\varExp'}$, the program
$\program{\varEnv}{\expFun{\varPat}{\varExp}}$ updates to
$\program{\varEnv'}{\expFun{\varPat}{\varExp'}}$.
%
%% the function
%% $\expFun{\varPat}{\varExp}$ updates to $\expFun{\varPat}{\varExp'}$ and the
%% environment updates from $\varEnv$ to $\varEnv'$.
%
Although directly updating closures in the output of a program will,
perhaps, be less common than other types of values, this rule is nevertheless
crucial for ``receiving'' changes propagated from elsewhere in the program.
%
%% the remaining derived rules.

%% The rule \ruleName{U-Var} says that, when updating its output value to
%% $\varVal'$, the environment $\varEnv$ updates to $\varEnv'$ which is like the
%% original except that $\varVar$ is bound to the new value $\varVal'$; the
%% expression $\varVar$ remains unchanged.

%% \parahead{Primitive Operations}

\subsubsection{Primitive Rules for Base Values}
\label{sec:prim-rules-base}

How to update primitive operations may vary in different deployments
of bidirectional evaluation.
The replacement and propagation rules are agnostic to these choices.
\autoref{fig:update-numbers} shows several simple primitive rules,
which we find useful in practice.
%
%% \autoref{fig:update-numbers} shows evaluation and update rules for addition.
%
For example,
there are two update rules, \ruleName{U-Plus-1} and \ruleName{U-Plus-2}, which,
respectively, re-evaluate the left or right operand ($\varExp_1$ or $\varExp_2$)
to a number ($\varNum_1$ or $\varNum_2)$ and then push back the updated
difference ($\varNum'-\varNum_1$ or $\varNum'-\varNum_2$) entirely to that
operand.
%
%% This ``top-down'' approach to inverting arithmetic operations is strictly less
%% expressive than the approach of solving value-trace equations~\cite{sns-pldi}.
%% %
%% We discuss the relationship between backpropagation- and trace-based approaches
%% in detail in \autoref{sec:discussion}.
%
Because there are two update rules, there are two valid program updates for
addition expressions.
%
%% Additional numeric primitive operations (not shown in
%% \autoref{fig:update-numbers}) are handled in similar fashion.
%
%% Some of the relational operator rules will be mentioned below.

%% \parahead{Function Application}

%% \parahead{Variable Binding Forms}

\subsubsection{Propagation Rules}

\begin{figure}[t]

\newcommand{\smallerHSepRule}
  %% {\hsepRule}
  %% {\hspace{0.16in}} %% max that allows first three rules
  %%                   %% to appear on one line
  {\hspace{0.12in}} %% smaller for arxiv-generated PDF

\small

\centering

$\inferrule*[lab=\ruleNameFig{U-Plus-1}]
  {
  {
  \startAlignedPremises{}
   \reducesToAligned{\varEnv}{\varExp_1}{\varNum_1}
   \\
   \updatesToAligned
     {\varEnv}
     {\varExp_1}
     {\varNum' - \varNum_1}
     {\varEnv_1}
     {\varExp_1'}
  \stopAlignedPremises{}
  }
  }
  {\updatesTo
    {\varEnv}
    {\expPlus{\varExp_1}{\varExp_2}}
    {\varNum'}
    {\varEnv_1}
    {\expPlus{\varExp_1'}{\varExp_2}}
  }
$
\smallerHSepRule
$\inferrule*[lab=\ruleNameFig{U-Plus-2}]
  {
  {
  \startAlignedPremises{}
   \reducesToAligned{\varEnv}{\varExp_2}{\varNum_2}
   \\
   \updatesToAligned
     {\varEnv}
     {\varExp_2}
     {\varNum' - \varNum_2}
     {\varEnv_2}
     {\varExp_2'}
  \stopAlignedPremises{}
  }
  }
  {\updatesTo
    {\varEnv}
    {\expPlus{\varExp_1}{\varExp_2}}
    {\varNum'}
    {\varEnv_2}
    {\expPlus{\varExp_1}{\varExp_2'}}
  }
$
\smallerHSepRule
$\inferrule*[lab=\ruleNameFig{U-Lt}]
  {
   \reducesTo{\varEnv}
             {\expBinop{\op{<}}{\varExp_1}{\varExp_2}}
             {\varBool}
  }
  {\updatesTo
    {\varEnv}
    {\expBinop{\op{<}}{\varExp_1}{\varExp_2}}
    {\lnot\varBool}
    {\varEnv}
    {\expBinop{\op{>=}}{\varExp_1}{\varExp_2}}
  }
$

\vsepRule

$\inferrule*[lab=\ruleNameFig{U-And-True}]
  {
  {
  \startAlignedPremises{}
   \updatesToAligned{\varEnv}{\varExp_1}{\expTrue}{\varEnv_1}{\varExp_1'}
   \\
   \updatesToAligned{\varEnv}{\varExp_2}{\expTrue}{\varEnv_2}{\varExp_2'}
  \stopAlignedPremises{}
  }
   \sepPremise
   \varEnv' = \mergeEnvsAllArgs{\varEnv}{\varEnv_1}{\varEnv_2}{\varExp_1}{\varExp_2}
  }
  {\updatesTo
    {\varEnv}
    {\expBinop{\op{\&\&}}{\varExp_1}{\varExp_2}}
    {\expTrue}
    {\varEnv'}
    {\expBinop{\op{\&\&}}{\varExp_1'}{\varExp_2'}}
  }
$
\hsepRule
$\inferrule*[lab=\ruleNameFig{U-And-False-1}]
  {
   \updatesTo{\varEnv}{\varExp_1}{\expFalse}{\varEnv_1}{\varExp_1'}
  }
  {\updatesTo
    {\varEnv}
    {\expBinop{\op{\&\&}}{\varExp_1}{\varExp_2}}
    {\expFalse}
    {\varEnv_1}
    {\expBinop{\op{\&\&}}{\varExp_1'}{\varExp_2}}
  }
$

%% \vsepRule
%% 
%% $\inferrule*[lab=\ruleNameFig{E-Lt}]
%%   {
%%    \reducesTo{\varEnv}{\varExp_1}{\varNum_1}
%%    \\
%%    \reducesTo{\varEnv}{\varExp_2}{\varNum_2}
%%    %% \\
%%    %% \varNum_1 < \varNum_2 = \varBool
%%   }
%%   {\reducesTo
%%     {\varEnv}
%%     {\expBinop{\op{<}}{\varExp_1}{\varExp_2}}
%%     %% {\varBool}
%%     {\varNum_1 < \varNum_2}
%%   }
%% $
%% %
%% \hsepRule
%% %
%% $\inferrule*[lab=\ruleNameFig{U-Lt}]
%%   {
%%    \reducesTo{\varEnv}
%%              {\expBinop{\op{<}}{\varExp_1}{\varExp_2}}
%%              {\varBool}
%%   }
%%   {\updatesTo
%%     {\varEnv}
%%     {\expBinop{\op{<}}{\varExp_1}{\varExp_2}}
%%     {\lnot\varBool}
%%     {\varEnv}
%%     {\expBinop{\op{>=}}{\varExp_1}{\varExp_2}}
%%   }
%% $

%% \vsepRule

\caption{Update for Primitive Operations (selected rules).}
\label{fig:update-numbers}
\end{figure}

How to update variables and variable binding forms is the heart of our
formulation; they are what allow changes to values at the leaves of
the program to flow throughout the program.
To simplify the presentation, the evaluation and update rules for let-bindings
and function calls assume only variable patterns rather than arbitrary ones, as in
our implementation.

\parahead{Variables}

When \ruleName{U-Var} updates the output of $\varVar$ to
$\varVal'$, the environment $\varEnv$ updates to $\varEnv'$ which is like the
original except that $\varVar$ is bound to the new value $\varVal'$; the
expression $\varVar$ remains unchanged.

\parahead{Let-Bindings}

Like \ruleName{E-Let}, the first premise of \ruleName{U-Let} evaluates the
expression $\varExp_1$ to a value $\varVal_1$, to be bound to $\varVar$
and added to $\varEnv$ when subsequently considering the expression body
$\varExp_2$.
Dual to \ruleName{E-Let}---which, in the new environment, evaluates $\varExp_2$ to a
value $\varVal_2$---the second premise of \ruleName{U-Let} pushes back the
expected value $\varVal_2'$, producing a
(potentially updated) expression body $\varExp_2'$ and
(potentially updated) environment
$\envCat{\varEnv_2}{\envBind{\varVar}{\varVal_1'}}$
that is structurally equivalent to the original environment
$\envCat{\varEnv}{\envBind{\varVar}{\varVal_1}}$
(their domains are equal).
Notice that, to produce this new program, the value $\varVal_1'$
bound to $\varVar$ may differ from the original value $\varVal_1$.
To discharge this obligation, the third premise pushes $\varVal_1'$
back to $\varExp_1$, producing a
(potentially updated) expression $\varExp_1'$ and
(potentially updated) environment $\varEnv_1$ that is structurally
equivalent to $\varEnv$.
Putting these pieces together, the new expression is
$\expLet{\varVar}{\varExp_1'}{\varExp_2'}$.
What remains is to reconcile $\varEnv_1$ and $\varEnv_2$ with the original
$\varEnv$.
The rules ensure that $\varEnv_1$ and $\varEnv_2$ are both structurally equivalent
to $\varEnv$, but each may have induced updates to one or more bindings in
$\varEnv$.
As demonstrated below, updated bindings may conflict---there
may be variables $\varVarY$ such that
$\varEnv_1(\varVarY)$ and $\varEnv_2(\varVarY)$ are different.
To produce the final environment, the conclusion of \ruleName{U-Let} uses an
environment merge operation,
$\mergeEnvsAllArgs{\varEnv}{\varEnv_1}{\varEnv_2}{\varExp_1}{\varExp_2}$,
discussed next.

\newcommand{\paraMerge}{Environment, Value, and Expression Merge}
\parahead{\paraMerge}

The merge operation---used by all rules that update multiple
subexpressions---takes two structurally equivalent environments $\varEnv_1$ and $\varEnv_2$ to merge,
as well as several optional arguments: the
original environment $\varEnv$, and the expressions $\varExp_1$ and $\varExp_2$
that were updated when $\varEnv_1$ and $\varEnv_2$ were produced.
We consider two implementations for this operation.

\begin{definition}[Conservative Two-Way Merge]

\emph{Two-way environment merge}
$\mergeEnvsTwoWay{\varEnv_1}{\varEnv_2}{\varExp_1}{\varExp_2}$,
reconciles bindings as follows (the subscript $\varEnv$ is omitted
because the original environment is ignored).
\begin{align*}
\mergeEnvsTwoWay
  {(\envCat{\varEnv_1}{\envBind{\varVar}{\varVal_1}})}
  {(\envCat{\varEnv_2}{\envBind{\varVar}{\varVal_2}})}
  {\varExp_1}
  {\varExp_2} &=
  {(\envCat{\varEnv'}{\envBind{\varVar}{\varVal}})}
  \textrm{ where } \varEnv' = \mergeEnvsTwoWay{\varEnv_1}{\varEnv_2}{\varExp_1}{\varExp_2}
  \textrm{ and } \varVal =
    \left\{
      \begin{array}{ll}
        \varVal_1 & \mbox{if } \varVal_1 = \varVal_2 \\
        \varVal_1 & \mbox{if } \varVar \notin \freeVars{\varExp_2} \\
        \varVal_2 & \mbox{if } \varVar \notin \freeVars{\varExp_1}
      \end{array}
    \right.
\end{align*}
\noindent
If neither environment updates the $\varVar$ binding, or if both update it
to the same new value, the first equation adds the new value to the
merged environment.
Otherwise, the second (resp. third) equation adds the new value
$\varVal_1$ from the left environment $\varEnv_1$ (resp. $\varVal_2$
from the right environment $\varEnv_2$) only if $\varVar$ does not
appear free in $\varExp_2$ (resp. $\varExp_1$).

Two-way merge is conservative; if a variable appears free in both expressions,
two-way merge requires that it be updated to the same new value in both
environments.
In other words, \emph{all} occurrences of an updated variable must be updated
consistently for the overall update to succeed.
Though restrictive, this would be necessary to ensure that the new program
evaluates to the updated value;
\autoref{sec:correctness} describes a correctness property that
depends on the use of two-way merge.
In practice, however, we prefer to support interactions where the
user can update only a subset of the uses---often just one---of a
particular variable.
To support this workflow, we propose an alternative way to merge
environments.

\end{definition}

\begin{definition}[Optimistic Three-Way Merge]

\emph{Three-way environment merge}
$\combineEnvs{\varEnv}{\varEnv_1}{\varEnv_2}$
performs \emph{three-way value merge} on variable bindings as follows
(the superscripts $\varExp_1$ and $\varExp_2$ are omitted
because these expressions are ignored).
\begin{align*}
\combineEnvs
  {(\envCat{\varEnv}{\envBind{\varVar}{\varVal}})}
  {(\envCat{\varEnv_1}{\envBind{\varVar}{\varVal_1}})}
  {(\envCat{\varEnv_2}{\envBind{\varVar}{\varVal_2}})} &=
  ({\envCat{\varEnv'}{\envBind{\varVar}{(\mergeVals{\varVal}{\varVal_1}{\varVal_2})}}})
  \textrm{ where } \varEnv' = \combineEnvs{\varEnv}{\varEnv_1}{\varEnv_2}
\end{align*}
The value merge operation $\mergeVals{\varVal}{\varVal_1}{\varVal_2}$
(not shown) recursively traverses the subvalues of three structurally
equivalent values, until the rule for base cases---for merging
constants---chooses $\varVal_2$ if it differs from $\varVal$ (even if
$\varVal_2$ and $\varVal_1$ conflict) and $\varVal_1$ otherwise.\footnote{
An implementation could, instead, choose to favor updates from the left,
or even propagate all combinations of choices.
}
Closure values include expressions, so we also define a \emph{three-way
expression merge} operation $\mergeExps{\varExp}{\varExp_1}{\varExp_2}$
(not shown) in similar fashion.
Three-way merge is optimistic; it succeeds despite conflicts.
As a result, use of three-way merge precludes the correctness property
enjoyed by the use of two-way merge.
Nevertheless, we expect three-way merge to be the default mode of use
in practice and choose it for our implementation
(\autoref{sec:implementation}).
In our experience, we find that desirable updates are often
produced by this configuration (\autoref{sec:evaluation}).

\end{definition}

\begin{example}

Consider the expression
$\expLetLong{\varVar}{\num{1}}{\expListTwo{\varVar}{\varVar}}$,
which evaluates to $\expListTwo{\num{1}}{\num{1}}$.
(The rule \ruleName{U-Cons} for updating lists, discussed
in detail below, updates the list element expressions and then merges the
updated environments.)
If updated to, say,
$\expListTwo{\num{1}}{\num{2}}$ or $\expListTwo{\num{0}}{\num{2}}$, the \ruleName{U-Let},
\ruleName{U-Cons}, and \ruleName{U-Var} rules---together with
the right-biased, three-way merge---combine to update the
program to $\expLetLong{\varVar}{\num{2}}{\expListTwo{\varVar}{\varVar}}$,
which, when evaluated, produces
$\expListTwo{\num{2}}{\num{2}}$.
Using two-way merge, however, update fails to produce a solution.

\end{example}

\parahead{Function Application}

Having seen the basic approach to propagating changes through
environments, we now turn to the rule \ruleName{U-App} for function
application (\autoref{fig:eval-update}).\footnote{
Consistent with the standard rewriting of
$\expLet{\varVar}{\varExp_1}{\varExp_2}$ to
$\expApp{(\expFun{\varVar}{\varExp_2})}{\varExp_1}$,
\ruleName{\footnotesize U-Let} is derivable from
\ruleName{\footnotesize U-Fun} and
\ruleName{\footnotesize U-App}.
}
The approach is like that of \ruleName{U-Let}, extended now to also
deal with the closure through which a value flows.
The first two premises re-evaluate the function $\varExp_1$ to a closure
$\closure{\varEnv_f}{\varVar}{\varExp_f}$ and the argument $\varExp_2$ to a
value $\varVal_2$.
The third premise pushes the updated value $\varVal'$ back through the
function call, specifically, through the function body $\varExp_f$, where the
closure environment is extended with the binding $\envBind{\varVar}{\varVal_2}$.
This produces a new function body $\varExp_f'$ and a new, structurally
equivalent environment $\envCat{\varEnv_f'}{\envBind{\varVar}{\varVal_2}'}$
with a new argument value $\varVal_2'$.
Three terms must be reconciled with the original program:
the function environment ($\varEnv_f'$),
the function body ($\varExp_f'$), and
the function argument ($\varVal_2'$).
The fourth premise pushes the first and second terms, in the form of the closure
$\closure{\varEnv_f'}{\varVar}{\varExp_f'}$, back to the original function
expression $\varExp_1$; the result is a new program
$\program{\varEnv_1}{\varExp_1'}$.
The fifth premise pushes the third term, $\varVal_2'$, back to the
original argument expression $\varExp_2$; the result is a new program
$\program{\varEnv_2}{\varExp_2'}$.
Environment merge is used to combine $\varEnv_1$ and
$\varEnv_2$, and the updated function application expression is
$\expApp{\varExp_1'}{\varExp_2'}$.

\parahead{Control-Flow}

Rule \ruleName{U-If-True} (\autoref{fig:eval-update}) pushes the
updated value back to the then-branch,
assuming that the same branch will be taken by the new program.
Rule \ruleName{U-If-False} (not shown) does the same for the else-branch.
Notice that new environment $\varEnv'$ is the result of merging the
original environment $\varEnv$ with the environment $\varEnv_2$ produced
when updating the then-branch.
The conservative two-way merge would ensure that all free variables of $\varExp_1$
must be bound to the same values in $\varEnv'$ as in $\varEnv$; thus,
the guard expression in the new program would evaluate to the same boolean value.
As discussed above, however, our implementation is configured to use
three-way merge.
Indeed, several of our example use cases for direct manipulation
interaction (\autoref{sec:evaluation}) purposely alter control-flow, \eg{}~because of a change to a boolean flag.

\begin{example}

Consider the expression
$\expApp
  {(\expFun{\varVar}
     {\expIteLong
        {\expBinop{==}{\varVar}{\num{1}}}
        {\varVar}
        {\num{3}}})}
  {\num{1}}$
which evaluates to $\num{1}$.
If the user updates the value to $\num{2}$, the change will be pushed back to
the then-branch, and then back through the variable use to the
function argument.
If using two-way merge, the update would fail because the updated variable,
$\varVar$, is free in the guard expression.
If using three-way merge, the updated expression would be
$\expApp
  {(\expFun{\varVar}
     {\expIteLong
        {\expBinop{==}{\varVar}{\num{1}}}
        {\varVar}
        {\num{3}}})}
  {\num{2}}$,
which, when evaluated, takes the else-branch and produces $\num{3}$
instead of $\num{2}$.

\end{example}

\parahead{Expression Freeze}

The update rules are applied ``automatically'' to all relevant (sub)expressions
when trying to reconcile the program with a new output value.
The $\expFreeze{\varExp}$ expression is semantically a no-op
(\ruleName{E-Freeze} in \autoref{fig:eval-update}) but is one
simple way to control the update algorithm, by requiring that the expression
$\varExp$ and values $\varVal$ it computes remain unaltered (\ruleName{U-Freeze}
in \autoref{fig:eval-update}).

\subsubsection{Primitive Rules for Lists and Dictionaries}
\label{sec:prim-rules-lists}

We have considered replacement, primitive, and propagation rules for a
core language with base values.
We now discuss primitive rules for updating data structures in
\littleLeo, namely, lists and dictionaries.

\parahead{List Construction}

\begin{figure}[t]

\small

\centering

$\inferrule*[right=\ruleNameFig{E-Cons}]
  {
  {
  \startAlignedPremises{}
   \reducesToAligned{\varEnv}{\varExp_1}{\varVal_1}
   \\
   \reducesToAligned{\varEnv}{\varExp_2}{\expListTwo{\varVal_2}{\cdots}}
  \stopAlignedPremises{}
  }
  }
  {\reducesTo
    {\varEnv}
    {\expCons{\varExp_1}{\varExp_2}}
    {\expListTwo{\varVal_1,\varVal_2}{\cdots}}
  }
$
\hsepRule
$\inferrule*[right=\ruleNameFig{U-Cons}]
  {
  {
  \startAlignedPremises{}
   \updatesToAligned{\varEnv}{\varExp_1}{\varVal_1'}{\varEnv_1}{\varExp_1'}
   \\
   \updatesToAligned{\varEnv}{\varExp_2}{\expListTwo{\varVal_2'}{\cdots}}{\varEnv_2}{\varExp_2'}
  \stopAlignedPremises{}
  }
   \sepPremise
   \varEnv' = \mergeEnvsAllArgs{\varEnv}{\varEnv_1}{\varEnv_2}{\varExp_1}{\varExp_2}
  }
  {\updatesTo
    {\varEnv}
    {\expCons{\varExp_1}{\varExp_2}}
    {\expListTwo{\varVal_1',\varVal_2'}{\cdots}}
    {\varEnv'}
    {\expCons{\varExp_1'}{\varExp_2'}}
  }
$

\vsepRule

$
\inferrule*[right=\ruleNameFig{U-List}]
  {\reducesTo
    {\varEnv}
    {\expList{\varExp_1}{\cdots}{\varExp_n}}
    {\varVal}
   \\
   \listDiff{}=\computeListDiff{\varVal}{\varVal'}
   \\
   \foldListDiffEquals
     {\varEnv}
     {\expList{\varExp_1}{\cdots}{\varExp_n}}
     {\listDiff{}}
     {\varEnv'}{\varExp'}
  }
  {\updatesTo
    {\varEnv}
    {\expList{\varExp_1}{\cdots}{\varExp_n}}
    {\varVal'}
    {\varEnv'}
    {\varExp'}
  }
$

\vsepRule

$\inferrule* %% [lab=\ruleNameFig{Apply-List-Diff-Empty}]
  { }
  {\foldListDiffEquals
     {\varEnv}{\valNil}{\emptyDiff}
     {\varEnv}{\valNil}
  }
$
\hsepRule
$\inferrule* %% [lab=\ruleNameFig{Apply-List-Diff-Keep}]
  {\foldListDiffEquals
        {\varEnv}
        {\varExp_2}
        {\listDiff}
        {\varEnv'}
        {\varExp_2'}}
  {\foldListDiffEquals
    {\varEnv}
    {\expCons{\varExp_1}{\varExp_2}}
    {\diffCons{\diffKeep}{\listDiff}}
    {\varEnv'}
    {\expCons{\varExp_1}{\varExp_2'}}}
$

\vsepRule

$\inferrule* %% [lab=\ruleNameFig{Apply-List-Diff-Delete}]
  {\foldListDiffEquals
    {\varEnv}
    {\varExp_2}
    {\listDiff}
    {\varEnv'}
    {\varExp_2'}}
  {\foldListDiffEquals
    {\varEnv}
    {\expCons{\varExp_1}{\varExp_2}}
    {\diffCons{\diffDelete}{\listDiff}}
    {\varEnv'}
    {\varExp_2'}}
$
\hsepRule
$\inferrule* %% [lab=\ruleNameFig{Apply-List-Diff-Insert}]
  {\foldListDiffEquals
    {\varEnv}
    {\varExp}
    {\listDiff}
    {\varEnv'}
    {\varExp'}}
  {\foldListDiffEquals
    {\varEnv}
    {\varExp}
    {\diffCons{\diffInsert{\varVal'}}{\listDiff}}
    {\varEnv'}
    {\expCons{\coerceToExp{\varVal'}}{\varExp'}}}
$

\vsepRule

$\inferrule* %% [lab=\ruleNameFig{Apply-List-Diff-Update}]
  {
   \updatesTo{\varEnv}{\varExp_1}{\varVal'}{\varEnv_1}{\varExp_1'}
   \\
   \foldListDiffEquals
     {\varEnv}
     {\varExp_2}
     {\listDiff}
     {\varEnv_2}
     {\varExp_2'}
   \\
   \varEnv' = \mergeEnvsAllArgs{\varEnv}{\varEnv_1}{\varEnv_2}{\varExp_1}{\varExp_2}
  }
  {\foldListDiffEquals
    {\varEnv}
    {\expCons{\varExp_1}{\varExp_2}}
    {\diffCons{\diffUpdate{\varVal'}}{\listDiff}}
    {\varEnv'}
    {\expCons{\varExp_1'}{\varExp_2'}}}
$

\caption{Evaluation and Update for Lists.}
\label{fig:update-cons}
\end{figure}

\autoref{fig:update-cons} shows a standard evaluation rule (\ruleName{E-Cons})
for list construction, and a corresponding update rule (\ruleName{U-Cons}) that
propagates changes to the head (resp. tail) value back to the head (resp. tail)
expression.
Notice that these rules preserve the structure of existing cons expressions;
we choose not to include structure-changing rules that add and remove cons
expressions because of the amount of ambiguity they would introduce.

\newcommand{\paraListLiterals}{List Literals: Pretty Local Updates}

\parahead{\paraListLiterals{}}

The primitive rules presented in \autoref{fig:update-numbers}
update expressions without altering their structure.
Even rules such as \ruleName{U-Lt}, which update the operator,
preserve the arity of the application.
This approach ensures a predictable class of ``small''
changes, but the same restriction applied to data structures would preclude
seemingly benign changes---\eg{}~updating the empty list expression $\valNil{}$ with
new value $\expListOne{\num{1}}$.

Our design includes the rule \ruleName{U-List}
(\autoref{fig:update-cons}) to allow insertion and deletion inside list literals
that appear in the program---we refer to this form of structural change as
\emph{pretty local} to emphasize its limited effect on the program structure.
We write $\expList{\varExp_1}{\cdots}{\varExp_n}$ as syntactic sugar for the
nested list construction expression
{\expCons{\varExp_1}{\expCons{\cdots}{\expCons{\varExp_n}{\valNil}}}} that
terminates with the empty list.
The helper procedure $\computeListDiff{\varVal}{\varVal'}$ takes the original
and updated list values and computes a \emph{value difference} $\listDiff$ (a
``delta''), in this case, a sequence of \emph{list difference
operations}---$\diffKeep$, $\diffDelete$, $\diffInsert{\varVal'}$, or
$\diffUpdate{\varVal'}$.
Our implementation of $\helperop{Diff}$ uses a dynamic programming approach
that, intuitively, attempts to preserve as many contiguous sequences from the
original list as possible.
We reuse the syntax of the evaluation update judgement for one that pushes back
value differences (rather than just values),
with the subscript $\helperop{Diff}$ to help distinguish them:
$\foldListDiffEquals{\varEnv}{\expList{\varExp_1}{\cdots}{\varExp_n}}{\listDiff{}}{\varEnv'}{\varExp'}$
computes the list literal $\varExp'$ that results from traversing the original
list literal and the difference operations, keeping, inserting, deleting, or
updating expressions as dictated by the difference.
%% %

\parahead{String, Records, and Dictionaries}

Evaluation rules (not shown) for
string concatenation $\expConcat{\varExp_1}{\varExp_2}$,
record literals $\expRecdDots{\varField_1}{\varExp_1}$, and
record extension $\expRecdExtend{\varExp}{\varField}{\varExp_f}$ are standard.
Dictionary values are built using primitive operators \op{empty}, \op{get},
\op{insert}, \op{remove}, and \op{fromList}.
Update rules for dictionaries work much the same way as those for lists,
based on \emph{dictionary difference operations} that are analogous to the
list difference operations, discussed above.
Update rules for records and record extension are also similar,
except that there are no insertions or deletions.
Update rules for concatenating strings and appending lists require a more
nuanced approach, as explained in \autoref{sec:leo-custom-update}.

\subsubsection{Correctness}
\label{sec:correctness}

We now describe two correctness properties that relate evaluation and
evaluation update.
The first correctness property is straightforward.

\begin{theorem}[EvalUpdate] \label{prop:eval-update}
If $\reducesTo{\varEnv}{\varExp}{\varVal}$,
then $\updatesTo{\varEnv}{\varExp}{\varVal}{\varEnv}{\varExp}$.
(\ie{}~If the program $\program{\varEnv}{\varExp}$ evaluates to $\varVal$,
pushing the same value back to the program does not change the program.)
\end{theorem}

Next, we define the notation
$\updatesToPessimistic{\varEnv}{\varExp}{\varVal'}{\varEnv'}{\varExp'}$
(with a check mark above the right arrow) to refer to the
{conservative} version of update that uses two-way environment
merge.
The second correctness property pertains only to this version.

\begin{theorem}[Conservative UpdateEval] \label{prop:weak-correctness}
If $\updatesToPessimistic{\varEnv}{\varExp}{\varVal'}{\varEnv'}{\varExp'}$,
then $\reducesTo{\varEnv'}{\varExp'}{\varVal'}$.
(\ie{}~If, when using two-way merge, pushing $\varVal'$ back to the program
$\program{\varEnv}{\varExp}$ produces $\program{\varEnv'}{\varExp'}$,
the new program evaluates to the updated value).
\end{theorem}

Recall that the two challenges for correctness of update are when
variables uses are not updated consistently and when control-flow
decisions deviate from their original directions.
The two-way environment merge operator is defined and used precisely
to curb these situations.
A detailed proof sketch is available in
\suppMaterials{}~\citep{sns-oopsla-full}.

\subsection{Customizing Evaluation Update}
\label{sec:leo-custom-update}

Because of the inherent expressiveness of the language, evaluation update
cannot provide all possible intended behaviors that users may desire.
Consider the common evaluation and update pattern, below, that is not
well-handled by the update algorithm.
\begin{flalign*}
&
\reducesTo
  {\varEnv}
  {\expAppTwo{\op{map}}{\varExp_f}{\expList{\varExp_1}{\varExp_3}{\varExp_4}}}
  {\expList{\varVal_1}
           {\phantom{{\varVal_2}\ttcomma\miniSepThree}\varVal_3}
           {\varVal_4\phantom{\ttcomma{\varVal_5}\miniSepThree}}}
\\
&
\updatesToFancyUnderbrace
  {\varEnv}
  {\expAppTwo{\op{map}}{\varExp_f}{\expList{\varExp_1}{\varExp_3}{\varExp_4}}}
  {\expListFive{\varVal'_1}{\varVal_2}{\varVal_3}{\varVal'_4}{\varVal_5}}
  {\varEnv'}
  {\expAppTwo{\op{map}}{\varExp_f'}{\expListFive{\varExp_1'}{\varExp_2}{\varExp_3}{\varExp_4'}{\varExp_5}}}
  {\mathrm{Desired,\ but\ unavailable,\ program\ repair}}
\end{flalign*}
The \helperop{Diff} operation computes the following alignment between the
original and updated values:
that $\varVal_1$ and $\varVal_4$ have been updated to $\varVal'_1$ and $\varVal'_4$, and
new values $\varVal_2$ and $\varVal_5$ have been inserted after (the
updated versions of) $\varVal_1$ and $\varVal_4$.
What would be desirable is an updated program of the form indicated above,
where $\varExp_f'$, $\varExp_1'$, and $\varExp_4'$ are updated because of the two updated function
calls $\expApp{\varExp_f}{\varExp_1}$ and $\expApp{\varExp_f}{\varExp_4}$, and
where the synthesized expressions $\varExp_2$ are $\varExp_5$ are passed to the function $\varExp_f'$,
ideally producing the inserted values $\varVal_2$ and $\varVal_5$.

Unfortunately, the evaluation update approach described so far
cannot synthesize repairs of the desired form above.
To understand why, consider the standard definition of \verb+map+:

\begin{center}
\codeSizeInText
\begin{BVerbatim}
letrec map f list = case list of [] -> []; x::xs -> f x :: map f xs
\end{BVerbatim}
\end{center}

\noindent
Notice that the original list value
$\expList{\varVal_1}{\varVal_3}{\varVal_4}$
is constructed
completely within the body of \verb+map+: non-empty (cons) nodes are created in
the \verb+list+ $=$ \verb+x::xs+ branch and the empty node is created in the
\verb+list+ $=$ \verb+[]+ branch.
To reconcile the updated list,
$\varVal_5$ would have to be inserted into the empty list \verb+[]+ in \verb+map+,
and $\varVal_2$ would have to be inserted into the cons-node.
Besides the fact that we disallow structural updates anywhere but \ruleName{E-List}
(\cf{}~the ``\paraListLiterals{}'' discussion in
\autoref{sec:leo-basic-update}), such
changes are not desirable---the new cons-node would not
be the result of applying \verb+f+ to anything; it would end up inserting the same
element in between \emph{all} elements in the output; and we want the definition
of \verb+map+, a library function, to be frozen anyway.

In short, the evaluation update has no means to provide simultaneous reasoning
about structural changes to list values and computations they pass through.
This is one simple programming pattern not handled well; surely there are many
others.

\subsubsection{User-Defined Update Functions}

\begin{figure}[t]

\small

\centering

$\inferrule*[right=\ruleNameFig{E-Lens}]
  {
   \reducesTo
     {\varEnv}
     {\expApp{\expProj{\varExp_1}{\ttfld{apply}}}{\varExp_2}}
     {\varVal}
  }
  {\reducesTo
    {\varEnv}
    {\expApplyLens{\varExp_1}{\varExp_2}}
    {\varVal}
  }
$
\hsepRule
$\inferrule*[right=\ruleNameFig{U-Lens}]
  {
   \reducesTo
     {\varEnv}
     {\varExp_2}
     {\varVal_2}
   \\
   \varVal_{3} =
%     \expRecdThree
     \expRecdTwo
       {\ttfld{input}}{\varVal_2}
       {\ttfld{outputNew}}{\varVal'}
   \\\\
   \reducesTo
     {\envCat
        {\varEnv}
        {\envBind{\varVar}
                 {\varVal_{3}}}}
     {\expApp{\expProj{\varExp_1}{\ttfld{update}}}{\varVar}}
     {\expRecd
        {\ttfld{values}}
        {\expList{\cdots}{\varVal_2'}{\cdots}}}
   \\\\
   \varVar\textrm{ fresh}
   \\
   \updatesTo{\varEnv}{\varExp_2}{\varVal_2'}{\varEnv'}{\varExp_2'}
  }
  {\updatesTo
    {\varEnv}
    {\expApplyLens{\varExp_1}{\varExp_2}}
    {\varVal'}
    {\varEnv'}
    {\expApplyLens{\varExp_1}{\varExp_2'}}
  }
$

\vsepRule
%% \input{fig-update-custom-helpers}
%% \begin{figure}[t]
%% 
%% \small
%% 
%% \centering

$\inferrule*[right=\ruleNameFig{E-Update-App}]
  {
   \reducesTo
     {\varEnv}{\varExp}
     {\expRecdThree{\ttfld{fun}}{\closure{\varEnv_1}{\varVar}{\varExp_f}}
                   {\ttfld{input}}{\varVal_2}
                   {\ttfld{outputNew}}{\varVal'}}
   \\
   \\\\
   S = \set{ \miniSepThree \varVal_2' \miniSepThree | \miniSepThree
       ( \updatesTo{\envCat{\varEnv_1}{\envBind{\varVar}{\varVal_2}}}
                   {\varExp_f}
                   {\varVal'}
                   {\envCat{\varEnv_1}{\envBind{\varVar}{\varVal_2'}}}
                   {\varExp_f} ) \miniSepThree}
   \\
   |S| = n
  }
  {\reducesTo
    {\varEnv}
    {\expApp{\op{updateApp}}
            {\varExp}}
    {\expRecd{\ttfld{values}}
             {\expList{S_1}{\cdots}{S_n}}}}
$

\vsepRule

$\inferrule*[right=\ruleNameFig{E-Diff}]
  {
  {
  \startAlignedPremises
   \reducesTo{\varEnv}{\varExp_1}{\varVal_1}
   \\
   \reducesTo{\varEnv}{\varExp_2}{\varVal_2}
  \stopAlignedPremises
  }
   \sepPremise
   \listDiff{} = \computeListDiff{\varVal_1}{\varVal_2}
  }
  {\reducesTo
    {\varEnv}
    {\expAppTwo{\op{diff}}{\varExp_1}{\varExp_2}}
    {\coerceToVal{\listDiff{}}}
  }
$
\hsepRule
$\inferrule*[right=\ruleNameFig{E-Merge}]
  {
  {
  \startAlignedPremises
   \reducesToAligned{\varEnv}{\varExp_1}{\varVal_1}
   \\
   \reducesToAligned{\varEnv}{\varExp_2}{\expList{\varVal_2}{\cdots}{\varVal_n}}
  \stopAlignedPremises
  }
   \sepPremise
   \varVal =
     \mergeVals{\varVal_1}{\varVal_2}{\mergeVals{\varVal_1}{\cdots}{\varVal_n}}
  }
  {\reducesTo
    {\varEnv}
    {\expAppTwo{\op{merge}}{\varExp_1}{\varExp_2}}
    {\varVal}
  }
$
%% \caption{Evaluation of Primitive Helper Functions for User-Defined Lenses.}
%% \label{fig:update-custom-helpers}
%% \end{figure}

\caption{Evaluation and Update for User-Defined Lenses and Primitive Helper Functions.}
\label{fig:update-custom}
\end{figure}

Instead of providing built-in support for updating \verb+map+ and
other common building blocks, we expose an API for users (or
libraries) to customize evaluation update.
Specifically, when a ``bare'' function is called, the program may
optionally provide a second ``\verb+update+'' function that specifies how
to push new argument values back to the call.
A pair comprising bare and \verb+update+ functions forms a
\emph{lens}~\cite{lenses},
specified in \littleLeo{} as a record with the following type
(written in comments starting with `\verb|--|' because, currently, our presentation and implementation do
not include types):

\vspace{2pt} %% HACK

\begin{center}
\codeSizeInText
\begin{BVerbatim}
-- type alias Lens a b =
--   { apply: a -> b, update: {input: a, outputNew: b} -> {values: List a} }
\end{BVerbatim}
\end{center}

\vspace{1pt} %% HACK

\noindent
In lieu of types,
the expression $\expApplyLens{\varExp_1}{\varExp_2}$ syntactically marks
the function application as a lens application.
The \ruleName{E-Lens} rule (\autoref{fig:update-custom})
projects the $\ttfld{apply}$ field of the lens
argument $\varExp_1$ and then applies it to the argument $\varExp_2$.
To push a new value $\varVal'$ back to the lens application
$\expApplyLens{\varExp_1}{\varExp_2}$, the \ruleName{U-Lens} rule uses the
$\ttfld{update}$ function of the lens.
The function argument is re-evaluated to $\varVal_2$
and, together with the new output $\varVal'$, is passed to the
$\expProj{\varExp_1}{\ttfld{update}}$ function.
Each value $\varVal_2'$ in the \verb+values+ list of results is pushed back to
the expression argument $\varExp_2$ and then used as the argument of the updated
function call expression.

Because the lens mechanism in \littleLeo{} is intended to provide a way to
customize the built-in update algorithm, we expose several internal
operators---\op{updateApp}, \op{diff}, and \op{merge}
(\autoref{fig:syntax})---that custom \verb+update+ functions can refer to.
We will explain the semantics of these operations (\autoref{fig:update-custom})
as they arise in the discussion below.

\subsubsection{Example Lenses}

Before defining a custom lens for list \verb+map+, we start with a simpler example:
mapping a ``\verb|MaybeOne|'' value, encoded as a list with either zero or one elements.

\parahead{Maybe Map}

\begin{figure}[t]
\codeSizeInFigure
\newcommand{\ttcomment}[1]{\texttt{\textit{#1}}}
\begin{Verbatim}
-- type alias MaybeOne a = List a

-- maybeMapSimple : (a -> b) -> MaybeOne a -> MaybeOne b
   maybeMapSimple f mx = case mx of []  -> []
                                    [x] -> [f x]

-- maybeMapLens : a -> Lens (a -> b, MaybeOne a) (MaybeOne b)
   maybeMapLens default =
     { apply (f, mx) = Update.freeze maybeMapSimple f mx
     , update {input = (f, mx), outputNew = my} =
         case my of
           []  -> { values = [(f, [])] }
           [y] ->
             let z = case mx of [x] -> x; [] -> default in
             let res = Update.updateApp {fun (g,w) = g w, input = (f,z), outputNew = y}
             in { values = map (\(newF,newZ) -> (newF, [newZ])) res }
     }

-- maybeMap : a -> (a -> b) -> MaybeOne a -> MaybeOne b
   maybeMap default f mx =
     Update.applyLens (maybeMapLens default) (f, mx)
\end{Verbatim}
\caption{Custom Lens for \texttt{MaybeOne.map}.}
\label{fig:map-maybe-lens}
\end{figure}

\autoref{fig:map-maybe-lens} shows a straightforward definition of
\verb+maybeMapSimple+, which is frozen to prevent changes to this ``library''
function.
When reversing calls to \verb+maybeMapSimple+, the built-in update algorithm
cannot deal with adding or removing elements from the argument list (as with
list \verb+map+, discussed above).
Thus, \autoref{fig:map-maybe-lens} defines a custom lens called
\verb+maybeMapLens+.

To deal with the case when the updated value includes an element when there
was none before, \verb+maybeMapLens+ is parameterized by a \verb+default+ element.
The functions \verb+apply+ and \verb+update+ take arguments \verb+f+ and
\verb+mx+ as a pair.
The \verb+maybeMap+ definition on the last line of \autoref{fig:map-maybe-lens}
is defined as the
application of this lens (wrapped in \verb+applyLens+) to its arguments
packaged up in a pair.
In the forward direction, the \verb+apply+ function
simply invokes \verb+maybeMapSimple+.
In the backward direction, the \verb+update+ function uses a record pattern to
project the \ttfld{input} and \ttfld{outputNew} fields
and handles two cases.
If the new output \verb+my+ is \verb+[]+, the updated \verb|MaybeOne| value
should be \verb+[]+, and the function \verb+f+ is left unchanged---these are
paired and returned as a singleton list of result \verb+values+.
If the new output \verb+my+ is \verb+[y]+, the goal is to push \verb+y+ back
through a call of \verb+f+.
If the original input maybe value \verb+mx+ is \verb+[x]+, then the function
call $\expApp{\ttf}{\ttz}=\expApp{\ttf}{\ttx}$ needs to be updated.
If the original input maybe value is \verb+[]+, however, there was no original
input; so, $\expApp{\ttf}{\ttz}=\expApp{\ttf}{\mathtt{default}}$ needs to be updated.

To achieve this, the primitive \op{updateApp} operator is used to push
\verb+y+ back through $\expApp{\ttf}{\ttz}$ using the built-in algorithm
(starting with rule \ruleName{U-App}).
The semantics of this operation (\ruleName{E-Update-App} in
\autoref{fig:update-custom}) computes all possible
updated values $\varVal'_2$ and puts them in a list returned to the source
program.
In this way, \verb+updateApp+ exposes the \ruleName{U-App} rule to custom
\verb+update+ functions.
Each value that comes out in \verb+results.values+ is a pair of a
possibly-updated function \verb+newF+ and possibly-updated argument \verb+newX+;
to finish, the second is wrapped in list and this pair forms a solution.
This function ``bootstraps'' from the primitive \ruleName{U-App} rule, lifting
its behavior to the \verb|MaybeOne| type.

For example, consider the function \verb|display [a, b, c] = [a, c + ", " + b]|
(essentially the function on line 14 of \autoref{fig:overview-initial-program})
and two calls to \verb|maybeMap defaultState display|,
where the definition
\verb|defaultState = ["?","?","?"]| serves as placeholder state data:

\vspace{2pt} %% HACK

\begin{center}
\codeSizeInText
\begin{BVerbatim}
maybeRow1 = maybeMap defaultState display [["New Jersey", "NJ", "Edison"]]
maybeRow2 = maybeMap defaultState display []
\end{BVerbatim}
\end{center}

\vspace{1pt} %% HACK

\noindent
Updating the result of \verb+maybeRow1+ to \verb+[]+ leads to updating
the argument to \verb+[]+.
Updating the result of \verb+maybeRow2+ to \verb+[["New Jersey", "Edison, NJ"]]+
leads to updating the argument to \verb+[["New Jersey", "NJ", "Edison"]]+.
Furthermore, updating the result of \verb+maybeRow2+ to
\verb+[["New Jersey", "Edison NJ"]]+ simultaneously inserts the appropriate
three-element list \emph{and} changes the separator \verb+", "+ to \verb+" "+.
None of these three interactions is possible if instead calling
\verb|maybeMapSimple display|, which is updated by the built-in algorithm alone.

\parahead{List Map and Append}

The \verb+maybeMapLens+ definition demonstrates an approach for dealing with updated
transformed values---pushing them back through function application, as
usual---and for dealing with newly inserted values---pushing them back through
function application with a default element.
We can extend this approach to a \verb+listMapLens+ definition that operates on lists with
arbitrary numbers of elements rather than just zero or one.
The definition (not shown) is a mostly
straightforward recursive traversal, with a few noteworthy aspects:
(1) the use of primitive operator \verb|diff| (for which \ruleName{E-Diff}
in \autoref{fig:update-custom} exposes the \helperop{Diff} operation used by
\ruleName{E-List}) to align the original and updated output lists;
(2) the use of primitive operator \verb|merge| (for which \ruleName{E-Merge}
exposes the three-way value merge operation) to combine multiple updates to
the input function; and
(3) when inserting a new element into the output list, choosing to use an
adjacent element from the original list (rather than a caller-specified default)
to push back through a function call.

Our library and examples implement several additional lenses for list
functions.
For example, we define a lens for appending lists, which generates multiple
candidate solutions when inserting elements at the ``split'' between the two
input lists.
(Our built-in evaluation update for concatenating strings does the same.)

\parahead{Control-Flow Repair}\label{sec:control-flow-repair}

As discussed in \autoref{sec:leo-basic-update}, our evaluation update
rules for conditionals---\ruleName{U-If-True} and \ruleName{U-If-False}---assume
that, after update, the predicate will evaluate to the same boolean and, thus,
the same branch will be taken.
Our treatment of conditionals does not update guard expressions, but
we can define a lens that simulates that effect.

The \verb|if_| function in \autoref{fig:fancy-if} employs a lens to augment
the built-in approach for updating if-expressions 
with the ability to change
the guard expression, by pushing a new boolean value back to it.
The lens simply wraps an if-expression in the forward direction,
but provides different behavior in the backward direction.
If the original guard \verb|c| evaluates to \verb|True| and the
original else branch \verb|e| evaluates to the new expected value \verb|v|,
then pushing \verb|False| back to \verb|c| constitutes a second
solution, called \verb|updateGuard|.
The treatment for when \verb|c|
evaluates to \verb|False| is analogous.
For example, \verb|if_ True 1 2| evaluates to \verb|1|; when pushing
\verb|2| back to the lens application, \verb|False| is pushed back to
the guard.

\begin{figure}[t]
\codeSizeInFigure
\begin{Verbatim}
   if_ cond thn els =
     Update.applyLens
       { apply (c, t, e) = if c then t else e
       , update {input=(c,t,e), outputNew=v} =
           let updateSameBranch =
             if c then (c, v, e) else (c, t, v)
           in
           let updateGuard =
             if (c && e == v) || (not c && t == v) then [(not c, t, e)] else []
           in
           { values = updateSameBranch::updateGuard }
       }
       (cond, thn, els)
   
   abs n =
    -- if  (n < 0)  then n  else (-1 * n)
       if_ (n < 0)       n       (-1 * n)
\end{Verbatim}
\caption{Custom Lens for Control-Flow Repair.}
\label{fig:fancy-if}
\end{figure}

\autoref{fig:fancy-if} shows an absolute value function, defined
in terms of \verb|if_| rather than \verb|if|.
The expression \verb|map abs (List.range -2 2)| evaluates to \verb|[-2,-1,0,-1,-2]|,
which is not what is intended; the desired fix would change the guard predicate
(the first argument) to \verb|(n >= 0)|.
Suppose the programmer updates the last element in the output (\verb|-2|,
produced by the else-branch \verb|(-1 * n)|) to \verb|2|.
Because the branch not taken (the then-branch, \verb|n|) produces
the desired value, \verb|True| is pushed back to \verb|(n < 0)|
(\cf{}~\verb|updateGuard|).
To resolve this update, the \ruleName{U-Lt} rule (\autoref{fig:update-numbers})
flips the relational operator (keeping the operands the same), thus
producing the new guard \verb|(n >= 0)|.

\section{\snsLeo{}: Direct Manipulation Programming for HTML}
\label{sec:implementation}

We implemented bidirectional evaluation in the 
\snsLeo{} \emph{direct manipulation programming system}~\cite{sns-pldi}.
In addition to the novel update algorithm, our new system supports
writing programs in \leo{}---an Elm-like language that extends
\littleLeo{} with several programming conveniences
(\url{http://elm-lang.org/})---for generating HTML output.
Users can edit the HTML output of the program using graphical
user interfaces, which trigger the evaluation update to reconcile
the changes.
Our implementation is written in a combination of Elm and
JavaScript, extending the implementation of \sns{} by \citet{sns-pldi}.
Our changes constitute more than 12,000 lines of code.
%wc -l HTML* Update* EvalUpdate* OutputCanvas.elm Results.elm LazyList.elm GroupStartMap.elm ValBuilder.elm ValUnbuilder.elm MissingNumberMethods.elm Native/outputCanvas.js Native/MissingNumberMethods.js
% - 1573 for the added lines in preludeleo.
% = 9787
% If we add the changes to ElmParser (3000L) we are above 10K,
%
The new system, \snsLeo{}~\snsVersion{}, is available at
\url{http://ravichugh.github.io/sketch-n-sketch/}.

Our implementation of program update is configured to use
three-way environment merge, favoring the flexibility to update
only some uses of a variable over the correctness guarantees
afforded by two-way merge (\autoref{sec:leo-basic-update}).
In this section, we describe:
optimizations and other enhancements to turn the evaluation update relation
into an algorithm suitable in a practical setting;
features in \leo{} to support programming practical applications;
and
a user interface for manipulating HTML output values and choosing program
updates.

\subsection{Enhancements for Program Update in Practice}
\label{sec:implementation-update}

Our implementation addresses several concerns necessary for the evaluation update
relation described in \autoref{sec:leo} to form the basis for a practical
algorithm.

\parahead{Optimization 1: Tail-Recursive Update}

A direct implementation of the update algorithm would result in a call
stack that increases with each recursive call.
Because the stack space in current browsers is relatively limited, this approach
leads to exceptions for many benchmarks, even relatively
small ones.
Because the heap space is usually less limited than stack space, we applied a
rewriting of the update procedure to continuation-passing style, which makes the
update procedure tail-recursive and, thus, compile to a JavaScript while loop.
This transformation is also compatible with another feature that returns a lazy list
of all solutions computed by the algorithm.
In the future, we could also use this transformation to repeatedly pause the
computation, so that it would not block the user interface (which is
single-threaded in JavaScript).

\parahead{Optimization 2: Merging Closures}

Three-way merging environments na\"ively---following the definition of
$\combineEnvs{\varEnv}{\varEnv_1}{\varEnv_2}$---would require exponential time.
Each closure in the environment refers to the prefix of the environment
(which might have been modified).
Hence, to compare closures, we need to compare their environments, and so on.
A simple, but crucial, optimization consists of merging bindings for only
those variables which appear free in the associated function bodies.

\newcommand{\paraDiffs}{Propagating and Merging Edit Differences}
\parahead{Optimization 3: \paraDiffs}

Another fundamental scalability issue is that the evaluation update judgement
propagates expected values $\varVal'$, even though large portions of $\varVal'$
may be identical to the original values $\varVal$.
Instead, our implementation computes an \emph{edit difference}
between $\varVal$ and $\varVal'$ that, together with those values,
serves as a compact but complete characterization of the changes.
For example, for numbers and booleans, the edit difference is a boolean flag
indicating whether the value has changed (\ie{}~whether \ruleName{U-Const} needs
to process this value).
For lists, the edit difference is a list of index ranges
associated with a number of insertions, a number of
removals, or an update based on a value difference---comparable to the list differences described in ``\paraListLiterals{}''
(\autoref{sec:leo}).
Edit differences for other types of values, for expressions, and for
environments are similar.
These edit differences are propagated through the evaluation update algorithm.

We also expose edit differences to user-defined lenses, so that they can benefit from
this optimized representation.
First, compared to the presentation of \ruleName{U-Lens} in
\autoref{fig:update-custom}, we also include the field \ttfld{outputOld} in the
record argument $\varVal_3$ to \verb+update+: its value $\varVal$ is the
original result of the function call
$\expApp{\expProj{\varExp_1}{\ttfld{apply}}}{\varExp_2}$.
The \verb+update+ function can choose to take \verb+outputOld+ into account when
returning its list of new argument \verb+value+s.
Furthermore, to take advantage of the optimized representation, the record
argument also contains a \ttfld{diffs} field that describes the edit differences
that turn \verb+outputOld+ into \verb+outputNew+.
The \verb+update+ function can optionally return a \verb+diffs+ field (in
addition to \verb+values+), in which case, the evaluation update algorithm can
continue to propagate changes using the optimized representation.\footnote{
Reasoning with \texttt{values} and \texttt{diffs} can be thought of as
``states'' and ``operations'', respectively, in the terminology of
synchronization, as explained by \citet{lenses}.}
In our library, we implement a \verb+foldDiff+ helper function and use it to
define edit difference-based versions of the reversible \verb+map+ and
\verb+append+ lenses described in \autoref{sec:leo-custom-update}.

\parahead{Usability: Whitespace and Formatting}

So that updated programs remain readable and conducive to subsequent
programmatic edits, our implementation takes care to insert and remove
whitespace in a way that respects the whitespace conventions of surrounding
expressions (\cf{}~the list highlighted in green in
\autoref{fig:overview-add-row}).
To achieve this, our abstract syntax tree explicitly records whitespace in
between expressions and concrete syntax tokens, and these are used to
determine how much whitespace to insert before and after newly created
expressions.

\subsection{Enhancements for Programming in Practice}
\label{sec:implementation-leo}

In addition to the constants, lists, and records presented in \littleLeo{}
(\autoref{fig:syntax}),
\leo{} also supports tuples, user-defined datatypes,
and value-indexed dictionaries with an
arbitrary number of bindings.
Our current implementation does not perform type checking, but a standard
ML-style type system is fully compatible with our approach %% language design
and is planned for future work.

Programs that generate HTML typically perform a large amount of string
processing and JavaScript code generation.
We briefly describe language extensions that facilitate such tasks.

\parahead{Regular Expressions}\label{sec:regex}

Our implementation provides two common regular expression operators.
The first operation, $\expRegexExtract{\mathit{re}}{\varStr}$, takes a regular
expression $\mathit{re}$ (as a string) and a string $\varStr$ to transform, and
optionally returns a list of all the groups of the first match of $\mathit{re}$
to $\varStr$.
The update semantics consists of taking a set of non-overlapping modified
groups---taken greedily from the right---and pushing them back to their original
place in the original string.
For example, $\expRegexExtract{\expStr{b(.)}}{\expStr{bab}}$ produces
$\expApp{\texttt{Just}}{\expListOne{\expStr{a}}}$.
If the result is updated to $\expApp{\texttt{Just}}{\expListOne{\expStr{x}}}$,
the string $\varStr$ is updated to $\expStr{bxb}$.

The second operation, $\expRegexReplace{\mathit{re}}{f}{\varStr}$, takes a
regular expression $\mathit{re}$, a function $f$, and a string $\varStr$ to
transform.
The function argument provides access to the match information, including the
index into the string, the subgroups and their positions, the global match, and
the replacement number.
The function uses this information to produce a string.
Interestingly, the final string after replacement is an interleaved
concatenation of strings that did not change and applications of the lambda to
the record associated to each match.
For example, in the string $\expStr{arrow}$, if we replace \expStr{(rr|w)} with
the function
$f=\expFun{\ttm}{\ \expIteLong{\expEqual{\expProjField{\ttm}{\ttfld{match}}}{\expStr{w}}}{\expStr{r}}{\expStr{rm}}}$,
we build an expression that looks like
$\expPlus{\expStr{a}}
{\expPlus{\expApp{f}{\expRecd{\ttfld{match}}{\expStr{rr}}}}
{\expPlus{\expStr{o}}
{\expApp{f}{\expRecd{\ttfld{match}}{\expStr{w}}}}}}$.
We use this expression both for evaluation and update.
For the latter, we first run the update procedure on this expression.
Then, in the environment, we recover an updated function $f'$.
To update the original string $\varStr$, we gather the information about the matches
that changed (including the subgroups) and apply them to $\varStr$.

Using the reversible \regexextract{} operation,
we are able to build a
\verb+String+ library with reversible variants of several common operations:
\verb|take|, \verb|drop|, \verb|match|, \verb|find|, \verb|toInt|,
\verb|trim|, \verb|uncons|, and \verb|sprintf|.

\parahead{Long String Literals}

Many languages allow string literals to refer to variables or expressions, which
are then \emph{expanded} (a.k.a. \emph{interpolated}).
Our implementation of \leo{} provides long string literals---distinguished by
triple double quotes and which may span multiple lines---that support string
interpolation of expressions (written $\expLongStr{@(\varExp)}$).
To further facilitate string processing tasks, we also allow variables to be
defined within long string literals
(written $\expLongStr{@\expLet{\varVar}{\varExp;}{\varStr}}$).

\parahead{Dynamic Code Evaluation}

As is common in \web{} programming, several of our examples use
a \emph{dynamic code evaluation} primitive, $\expEval{\varExp}$,
to dynamically compute strings that are meant to parse and evaluate
as \leo{} expressions.
The evaluation and update rules (not shown) are straightforward; the former
employs the parser to convert a code string value $\varStr$ into an expression
to evaluate, and the latter additionally employs the unparser to push an updated
code string $\varStr'$ back to the expression that generated it.

\parahead{HTML-to-String Lens}

We designed a lens for parsing an HTML string to a list of
\leo{}-encoded HTML nodes.
Challenges for this implementation include:
tolerating a variety of malformed documents (as most practical HTML parsers do);
and
carefully tracking whitespace, quotation marks, and other characters that are
not stored in the resulting DOM (these details are needed to
respect the formatting conventions of the program).
As a result, for convenience, users can copy-and-paste HTML strings into long
string literals.

\subsection{Direct Manipulation User Interface for Updating HTML Output Values}
\label{sec:implementation-ui}

Our last major extension to \snsLeo{} is the user interface for updating
output values and interacting with the program update algorithm.
Below, we describe several different direct manipulation value editors.
Regardless of which value editor is used to make changes, the connection
to the update algorithm is as described in \autoref{sec:overview}
(\cf{}~``\paraComputingDisplaying{},'' ``\paraAmbiguity{},'' and
``\paraAutoSync{}'').

\parahead{Value Editors}

We implement three kinds of user interfaces for manipulating output.
The first mode is a \emph{Graphical User Interface}, which allows the user to make
edits directly in the HTML-rendered output.
Currently, the main edits we support are text-based:
in our translation of HTML text nodes, we add the \expStr{contenteditable}
attribute to allow changes to the text.
In the future, we could add direct manipulation widgets for common
properties of other kinds of elements, such as color, position, size, padding, \etc{}
The second mode is a \emph{Text Interface}, which allows the user to make
edits to the output value rendered as a string.
The interface allows the string to be rendered either as ``raw'' HTML or in the
syntax of \leo{} values.
The final mode integrates with the built-in \emph{DOM Inspector} provided by
modern \web{} browsers.
The features provided by the browser allow users to, for example, select DOM
elements---either by right-clicking or by navigating in a
separate view of the DOM tree---and then use built-in text- and GUI-based
panels for adding, removing, and editing elements and their attributes.

\section{Examples and Experiments}
\label{sec:evaluation}

To validate our approach, we implemented \benchmarksnum{} example programs in
\snsLeo{}---comprising approximately \benchmarkslocfloored{} lines of \leo{} code in total---that are
designed to facilitate a variety of useful direct manipulation interactions
enabled through bidirectional evaluation.
\autoref{fig:benchmarks} summarizes our examples and experiments.
Examples marked with asterisks have accompanying videos on the \web{}.
Next, we describe noteworthy aspects of the example programs.
Then, we report results from performance experiments for the update algorithm.

\subsection{Examples}
\label{sec:evaluation-examples}

We describe several example programs and corresponding direct
manipulation interactions.

\newcommand{\benchmarkName}[1]{#1}

\parahead{States Table (Overview Example)}

The \benchmarkName{States Table A} benchmark in
\autoref{fig:benchmarks} includes direct
manipulation text and DOM edits (like those described in
\autoref{sec:overview-text-edit} and \autoref{sec:overview-inspector-edit}),
and the \benchmarkName{States Table B} benchmark
corresponds to interactions with custom buttons
(like those in \autoref{sec:overview-add-row}).
Although the implementation details are not crucial,
the main takeaway is that custom user interface features can be built by:
(i) defining a lens that, in the forward direction, attaches
extra ``state'' to some data and, in the backward direction, refers to the
updated state to determine how to update the data; and
(ii) exporting HTML elements that store the state
and handle events---using JavaScript code generated as \leo{}
strings---that update the state in response to browser events.
\parahead{Scalable Recipe Editor}

A recipe is presented in
such a way that ingredient amounts can be scaled easily with respect to a
desired number of \verb|servings|.
The source of the recipe is stored as a string containing HTML code.
There, every occurrence of \expStr{multdivby(p,q)} is first replaced using regexes  (\autoref{sec:regex}) by the number \verb|(p/q)*servings|, where \verb|servings| is defined for the entire recipe.
The resulting string is then evaluated by a string-to-HTML lens.
To insert the quantity ``5 eggs'' proportional to a current number of \verb|servings| of 10, users can simply enter \expOutputStr{\_5\_ egg} in the output, and the \expOutputStr{\_5\_} is replaced by custom lenses to \expStr{multdivby(5,10)} in the source text.
Similarly, inserting \expOutputStr{\_5s\_} inserts a conditional plural in \expStr{s}.
Because all proportional quantities are connected to \verb+servings+ through
invertible arithmetic operations, the user can edit any of the values as
desired---\eg{}~to scale the recipe to make 32 servings, or to find how many
servings can be made with 12 eggs---all others are updated accordingly.

\parahead{Mini Markdown-to-HTML Editor}

We ported a regular expression-based PHP program that converts
Markdown strings to HTML strings.\footnote{
Markdown: \url{https://daringfireball.net/projects/markdown/};
PHP program: \url{https://gist.github.com/jbroadway/2836900}
}
We can, for example,
demarcate a string in the output text with underscores
that get pushed back to the Markdown string.
Then, after evaluation, the text is italicized due to \verb+<em>+
tags
inserted by regular expression transformations.
For more advanced functionality, we implemented lenses to:
translate Markdown headers (\verb|#|, \verb|##|, \etc{})
to their HTML counterparts (\verb|<h1>|, \verb|<h2>|, \etc{});
translate unordered and ordered list elements (\eg{}~\verb|<li>| to either
\expStr{* A} or \expStr{1. A}); and
translate \verb|<div>| and \verb|<br>| elements to the correct number
of newlines.

\parahead{Additional Examples}

Common to all examples is that text---text elements, links, buttons,
placeholders, attributes, \etc{}---can be changed from the output.
Here, we briefly describe the remaining examples.
\newcommand{\examplequick}[1]
  %% {\textbf{#1}}
  {{#1}}
\examplequick{Budgetting} is the computation of a budget for which, if we update
the surplus \verb|(income - expenses)| to be zero, program
updates include all choices for changing the cost values of \verb+lunch+,
\verb+registration+, or incomes such as \verb+sponsors+.
\examplequick{Model-View-Controller (MVC)} demonstrates an interactive page that
manipulates the state of the application with buttons and user-defined
functions.
In \examplequick{Linked-Text}, users can create links
(``variables'') between portions of text so that updating any clone
updates them all.
\examplequick{Dixit} is a game scoresheet to track bets and compute scores.
\examplequick{Translation} is an instruction manual in two languages where users can change the language, add, and clone translations.
\examplequick{\LaTeX{} in HTML} supports writing programs in a miniature
\LaTeX{} subset, including \verb|\newcommand|, \verb|\section|, \verb|\ref|,
\verb|\label|, and simple math commands such as \verb$\frac$.
An interesting lens for this example allows reference numbers in the output
to be updated and pushed back to corresponding reference names in the
\LaTeX{} program.

\subsection{Performance of Update Algorithm}

\newcommand{\nudgeRight}
  %% now that we're not showing speedups for Fastest/Slowest Upd,
  %% things need to be nudged a bit to the right.
  %% to be inserted by the script that generates: fig-benchmarks-data.tex
  {\hspace{0.09in}}

% Without ambiguity
\newcommand{\tableRow}[9]{
  #1 & #2 & #3 & #4 & #6 & #9 & #7 \\
}
\newcommand{\tableRowToUpdate}[9]{
  \rkc{#1} & #2 & #3 & #4 & #6 & #9 & #7 \\
}

\newcommand{\summaryRow}{
     {\total{Total} / \average{Average}}
  &  {\total{\benchmarksloc}}
  &  \average{\benchmarksevalaverage}
  &  \average{\benchmarksevalstddev}
%  &  {\total{\benchmarkssessiontime}}
  &  \vspace{-6pt}\total{\benchmarksnumupd}
  &  {\average{\benchmarkssolutionsaverage}}
  &  {\total{}} &  {\total{}}
  &  {\total{}} &  {\total{}}
  & \average{(\benchmarksaverageoptupd{}}
  & \average{\benchmarksaveragestddev{})}
  & \average {(\benchmarksaveragespeedup{})}
}

\newcommand{\columnHeader}[1]
  {\textbf{#1}}

\newcommand{\hasvideo}{*}
\newcommand{\speedup}[1]{#1{}$\times$}
\newcommand{\always}[1]{always #1}
\newcommand{\fromto}[2]{#1 to #2}
\newcommand{\noambiguity}{N/A}
\newcommand{\nospeedup}{N/A}
\newcommand{\total}[1]{\textit{#1}}
\newcommand{\average}[1]{\textbf{#1}}
\newcommand{\lessthanms}[1]{$<${}#1{}ms}
\newcommand{\plusminus}{$\pm$}
\begin{figure}[t]
%% \footnotesize
\small
%%%%%%%%%%%%%%% 1 2 3 45 6 7 8
\begin{tabular}{l|c|r@{\hspace{0pt}}l|c|c|r@{\hspace{0pt}}l|r@{\hspace{0pt}}l|r@{\hspace{0pt}}l@{\hspace{2pt}}c}
\columnHeader{Example} &
\columnHeader{LOC} &
\multicolumn{2}{c|}{\columnHeader{Eval}} &
%\columnHeader{Time} &
\columnHeader{\#Upd} &
\columnHeader{\#Sol} &
\multicolumn{2}{c}{\columnHeader{Fastest Upd}} &
\multicolumn{2}{|c}{\columnHeader{Slowest Upd}} &
\multicolumn{3}{|c}{\columnHeader{Average Upd}}
\\
\hline
\benchmarks
\hline
\summaryRow
\end{tabular}

\caption{
\columnHeader{LOC}: Lines of code in \leo{};
\columnHeader{Eval}: Time to evaluate program (in milliseconds) before any
direct manipulation changes;
\columnHeader{\#Sol} is the average of number of solutions obtained during interaction.
\columnHeader{Fastest/\mbox{Slowest}/Average Upd}: Of the \#Upd invocations of program
update, the fastest/slowest/average time taken by the optimized version (in milliseconds)---averaged from 10
trials---with the speedup against an average of 10 trials of the unoptimized.
Asterisks mark examples for which screencast videos are available.
}

\label{fig:benchmarks}
\end{figure}

\newcommand{\columnHeaderText}[1]
  {``{#1}''}

To validate that our program update algorithm is fast enough to support an
interactive direct manipulation workflow, we measured the running time for
several benchmarks.
\autoref{fig:benchmarks} shows a summary of our results.
Each benchmark consists of an example program and an \emph{interactive
editing session}.
The \columnHeaderText{LOC} column shows the number of \leo{} lines of code for
the initial program and \columnHeaderText{Eval} shows the running time (in
milliseconds) averaged over 10 trials.
For each example, we performed a series of direct manipulation edits and program
updates---each session produced a sequence of calls to the update algorithm.
\columnHeaderText{\#Upd} shows the
number of calls to the program update algorithm during the session.
The interactive sessions for programs marked with asterisks were recorded and
are available on the \web{}.

We conducted an offline performance evaluation by replaying the sequence of updates in each session.
We ran our benchmarks on Node.js 6.9.5 under Windows 10 running on an Intel(R) Core(TM) i7-6820HQ CPU
@ 2.70GHz with 32 GB of RAM, allocating 4GB of RAM and one of the 8 processors because JavaScript is single-threaded.
For each call to program update, we measured the time to compute
solutions with an unoptimized version of the algorithm---which includes Optimizations 1 and 2 described in
\autoref{sec:implementation-update}---and a ``fully-optimized'' version---which also includes Optimization 3 regarding
edit differences.
Note that without Optimizations 1 and 2, the
algorithm runs out of stack or heap stack on most benchmarks.
We performed each of these calls 10 times; the running times in the last three
columns of \autoref{fig:benchmarks} are averages over the 10 trials.
The \columnHeaderText{Fastest Upd} column shows the (average) running time of
the fastest call to update (using the optimized algorithm) for the
given session; \columnHeaderText{Slowest Upd} shows the slowest;
\columnHeaderText{Average Upd} shows the (average) running time off all calls in
the session, with the speedup of running time between optimized and unoptimized.

\parahead{Results}

The data in \autoref{fig:benchmarks} lends support to three observations.

\paragraph{Edit Difference Optimization is Crucial for Performance}
Consider the \columnHeaderText{Average Upd} column of the last row;
these averages are in parentheses to indicate that they are
averages across calls to update, as opposed to averages of the rows
above. 
Across all \benchmarksnumupd{} calls to update across all benchmarks,
the average running time for the fully-optimized algorithm is
\benchmarksaverageoptupd{}ms.
This is a \benchmarksaveragespeedup{} speedup compared to the unoptimized version.
Thus, the use of edit differences, rather than plain values, 
is crucial for making evaluation update feasible in our setting.

\paragraph{Performance of Evaluation Update is Similar to Evaluation}

The average evaluation update time (\benchmarksaverageoptupd{}ms)
is nearly the same as the average evaluation time (\benchmarksevalaverage{}ms).
Because the evaluation update algorithm performs much the same work as
evaluation, this suggests that our optimizations achieve most opportunities
for speedup.
Further gains, both for evaluation and update, are likely to result from
optimizing the interpreter---or compiling to ``native'' JavaScript code---as
opposed to additional optimizations of the current approach.
Extending evaluation update to the setting of compiled code is a direction for
future work.

\paragraph{Little Ambiguity in Our Example Interactions}

Across all \benchmarksnumupd{} calls to update across all benchmarks,
the average number of solutions is \benchmarkssolutionsaverage{}.
The degree of ambiguity for program repairs is heavily dependent on the
programs and interactions under consideration, so this number should not be
interpreted too broadly.
However, we argue that our example programs and interactions demonstrate a
variety of useful and realistic scenarios for interactive editing.
Together with the data, this suggests that experts can develop programs
in such a way that direct manipulation edits lead to the desirable repairs
without an overwhelming amount of ambiguity.

\newcommand{\hobit}{\textit{HOBiT}}

\newcommand{\biX}{\textit{Bi-X}}
\newcommand{\vuX}{\textit{Vu-X}}

\newcommand{\phpQuickFix}{\textit{PHPQuickFix}}
\newcommand{\phpRepair}{\textit{PHPRepair}}

\section{Related Work}
\label{sec:discussion}

The motivations and approach of our work overlap with various efforts towards
bidirectional programming,
automated program repair, and
combining programming languages with direct manipulation user interfaces.

\parahead{Bidirectional Programming}

\emph{Lenses}~\citep{lenses} have been an effective way to build
bidirectional transformations in a variety of domains, including
relational data~\cite{relational-lenses},
semi-structured data~\cite{lenses,bixid},
strings~\cite{boomerang,matching-lenses}, and
graphs~\cite{Hidaka:2010}.

Our work involves two notions of ``bidirectionality.''
First, we seek to reverse \emph{all} programs in a
general-purpose language; that is, we define a ``backward interpreter.''
Second, we allow user-defined functions
to customize the behavior of the backward interpreter,
by exposing its
operations in an API inspired by lenses.

\paragraph{Round-Trip Laws}
The foundational work on bidirectional programming
requires that lenses satisfy various \emph{round-trip laws}.
In contrast, our approach is simply for
users to write arbitrary pairs of (well-typed) \verb+apply+ and
\verb+update+ functions.
Many of the lenses we write to achieve custom user interface
interactions (\cf{}~\autoref{sec:evaluation-examples}) violate even
the basic laws, by introducing extra state that is unconditionally ``reset''
in the reverse direction.
We validate our more practically motivated design choices by
demonstrating a variety of desirable interactions.
In future work, static and dynamic mechanisms for checking round-trip laws
could be incorporated for situations in which programmers wish them to
be enforced.

\paragraph{Alignment}
Our update algorithm uses a \helperop{Diff} operation
based on a single heuristic, and this operation is exposed to user-defined
lenses through the \verb|diff| primitive.
For example, given a list \verb|[a,b]| updated to \verb|[a,c,b']|,
the alignment computed by \helperop{Diff}
says that \verb|b| is updated to \verb|c| and that \verb|b'|
is a new element inserted at the end.
However, aligning \verb|b| and \verb|b'|, and treating \verb|c| as an
insertion, may be preferable in a particular setting.
Furthermore, nested differences are not supported by \helperop{Diff}.
For example, if \verb|[x,y,z]| is updated to \verb|[x, ["b",[],[y]], z]|,
alignment fails because the expression which produced \verb|y| is assumed to
be updated with \verb|["b",[],[y]]|.
Instead, that expression should be
updated with \verb|y| and then propagated upwards.
In future work,
it would be useful to integrate alternate alignment mechanisms---and expose
these choices to user-defined lenses---as in the \emph{matching lenses}
framework proposed by \citet{matching-lenses}.

\paragraph{Alternatives to Lens Combinators}

The ``point-free'' combinator-style programming model---which prohibits giving
names to the results of intermediate computations---can pose significant
usability challenges.
One line of recourse is to forgo programming transformations at all, instead
synthesizing them from specifications of the desired input and output data
formats and examples connecting the two; \citet{Miltner:2018} and
\citet{Maina:2018} present such techniques for useful classes of
bidirectional string transformations.
Several other approaches aim to provide programming conveniences for
writing bidirectional transformations.

\emph{Bidirectionalization} aims to automatically
derive backward-functions for programs written with fewer syntactic
restrictions.
\emph{Syntactic bidirectionalization} approaches---which inspect the
syntax of function definitions---have been developed for domain-specific
languages, including a first-order, affine, treeless
language~\citep{Matsuda:2007} and a graph transformation
language~\citep{Hidaka:2010}.
The \emph{semantic bidirectionalization} approach of
\citet{Voigtlander:2009}---which inspects only the types of function
definitions---derives backward-functions for polymorphic functions written
in the general-purpose programming language Haskell.
This approach relies on the fact that polymorphic functions can
discriminate data structures but not the values contained within:
evaluation is instrumented with index information, later used to align
updated values with the originals.
This approach can handle a variety of examples, as long as output changes
preserve the shape of the original output data.

\citet{Matsuda:2015,Matsuda:2018} propose techniques that bring lens
programming closer to an unrestricted style of functional programming.
With \emph{applicative lenses}~\citep{Matsuda:2015}, lenses are lifted into
\emph{lens functions}, which can be manipulated with familiar higher-order
programming constructs.
Although written in a general-purpose functional language, programs must
explicitly manipulate lens functions in an applicative style (rather than as
plain ``unlifted'' functions).
Furthermore, although explicit lambdas---which introduce names (\ie{}~``points'')
for the results of intermediate computations---are allowed, variable uses are
restricted: when duplicating a value (by using a variable more than once),
each copy must be wrapped with a tag that specifies whether or not it is
relevant for subsequent updates.

In \hobit{}~\citep{Matsuda:2018}, bidirectional transformations can be
specified as unlifted functions.
To achieve this, they stage evaluation of a surface expression into a
partially evaluated \emph{residual expression}, which contains no function
applications in evaluation positions.
Given an environment that binds its free variables, a second ``get-evaluator''
reduces the residual expression to a value.
A ``put-evaluator'' pushes an updated value back to the
residual expression, producing an updated environment.
The latter two evaluators form a lens, by viewing the residual expression as a
forward-function and the original environment as its input.
Akin to the aforementioned approach of tagging variable uses, \hobit{} employs
an environment weakening operation to tolerate updates to variables that do
not appear free in the expression being updated---two-way environment merge
plays a similar role in our approach (\cf{} ``\paraMerge{}'' in
\autoref{sec:leo-basic-update}).
Unlike our approach, however, the \hobit{} backward-evaluator is limited to
first-order values, and the resulting changes---to the environment---are not
pushed back to the surface program.
Our approach eliminates the distinction between surface and residual
expressions, so that \emph{all} expression forms---including function
application---in a general-purpose language benefit from bidirectionality.
The result is that both data and code can be smoothly updated within our system.

Two additional aspects of \hobit{} are noteworthy.
One is its treatment of control-flow: each branch of a \verb+case+ expression
is equipped with an \emph{exit condition} and \emph{reconciliation function}
to support ``branch switching'' as in the approach of \citet{lenses}.
The expectation is that, in practice, branching decisions will often differ
between forward- and backward-evaluation.
In contrast, our examples mostly exercise updates which preserve branching
decisions (though \autoref{sec:control-flow-repair} shows how some branch
switching can be supported in our approach).
The second noteworthy aspect is the inclusion of an \verb+appLens+ operation
(similar to our \verb+applyLens+) to allow new primitive lenses to be defined.
For the class of bidirectional transformations that can be programmed in both
\hobit{} and \sns{}, it would be interesting to perform detailed case studies
in future work.

\paragraph{Applications to Documents and \Web{} Applications}
Several authors have employed bidirectional transformations
to develop structured documents and \web{} applications.

\citet{Hu2008} design a programmable editor for tree-structured documents,
where tree transformations are defined in an invertible language
of (injective) functions~\citep{Mu2004a,Mu2004b}.
Several aspects of their work are noteworthy in relation to ours.
One is the presence of a duplication operator.
To restore the equality invariant between copies of duplicated
data, the reverse semantics for this operation gives precedence to the value
of an updated copy; the process fails if copies are updated inconsistently.
Our three-way merge operation similarly gives precedence to updated uses of a variable,
but our default choice is to allow conflicting updates (\cf{} ``\paraMerge{}''
in \autoref{sec:leo-basic-update}).
Second, their support for duplication precludes the round-trip laws of
\citep{lenses}.
Instead, they prove a weaker ``stability'' law that, intuitively, says
that only one update needs to be performed per user edit.
Third, their system tracks \emph{edit tags} which demarcate inserted,
deleted, and updated values.
To further optimize our update algorithm, our user interface could track
edit tags and then push back edit differences only to modified sub-values,
rather than the entire output value (\cf{} ``\paraDiffs{}'' in
\autoref{sec:implementation-update}).

\citet{rajkumar_lenses_2014} define lens combinators in Haskell for
creating HTML \emph{form lenses} that push updated data (including insertions and deletions)
back to the program.
As with the approach of \citet{Matsuda:2015}, this approach requires
programming explicitly with lenses to obtain bidirectionality, and updates
are limited to data.

Within this category of work, the goals of \citet{nakano_consistent_2009}
are closest to ours: to provide a programming system in which both data and
code can be updated through direct manipulation GUIs.
In their system, \vuX{}%
, programs are
written in \biX{}, a bidirectional XQuery-like language for
transforming XML databases into HTML pages~\cite{Liu2007}.
Compared to the language used by \citet{Hu2008}, \biX{} is more
expressive, including support for variable binding, multi-argument
functions, and paths for addressing nodes in an XML document.
Regarding their language,
user-defined functions are limited to the top-level of
the program, and transformations on data structures (\eg{}~\verb+map+)
are primitive.
In contrast, our approach supports a general-purpose, higher-order
language and allows transformations of lists,
records, and user-defined data structures to be customized.
Regarding their user interface, there are two editing modes in \vuX{},
one for updating content and one for editing code.
In the former mode, as in our system, data and style values can be directly
manipulated, triggering synchronization and re-evaluation.
In the latter mode, a separate \emph{code builder} interface allows
the user to interactively change the structure of the program, using
concrete sample data to aid the development process.
Our system does not currently provide any user interface support for
changing the structure of the program; this would be useful to pursue
in future work.
We will also want to expose different editing modes for
different users; this may be useful even without a code builder, as some
constants (\eg{}~the separator string on line 14 of \autoref{fig:overview-initial-program})
may be thought of as code rather than data.
Similar to when using three-way merge in our approach,
their approach does not provide a strong correctness
property, because user updates may influence other parts of the
output due to duplication.
Finally, \vuX{} provides some support for distributed editing and access
control.
Providing such mechanisms in future work is needed to
truly allow a wide range of users---from ``expert programmers''
to ``designers'' to ``end users''---to interact with and modify
the application, as permissions allow, within the same system.

\parahead{Program Repair}

\citet{Monperrus2018} defines automatic repair as ``the transformation of an
unacceptable behavior of a program execution into an acceptable one according to
a specification'' and provides a comprehensive bibliography of static and
dynamic repair techniques.
Most relevant among these techniques are two that repair PHP programs based on
changes to HTML output.

\citet{samimi_automated_2012} present two tools: a first-pass static analysis
tool, \phpQuickFix{}, that repairs individual print statements which produce
malformed HTML; and a second-pass dynamic analysis tool, \phpRepair{}, that
instruments the evaluation of print statements on a given test suite.
Using these print traces, together with the expected HTML output for each input
test, \phpRepair{} generates string constraints where solutions correspond to
the addition, removal, and modification of print statements to satisfy the input-output
behavior of the test suite.
Although limited to modifying constant string arguments to print statements (but
not strings that flow through primitive operations, variables, and function
calls), the \phpRepair{} achieves good results in practice.

The approach of \citet{Wang2012} performs dynamic taint analysis of strings with
finer granularity than \phpRepair{}.
When the user makes a change to the HTML output (for a single run, unlike
multiple runs as in \phpRepair{}), the string trace is used to repair
string constants in the program.
Then they perform a combination of static and dynamic analysis to determine the
impact of the repair on the output.
This analysis may conclude that automated repair is not possible (\eg{} due to
ambiguity, effects on unrelated parts of the output, or because the origins of
the transformed string are not in the source program), in which case the user is
prompted to intervene.
The trace-based approach in earlier versions of \sns{}~\citep{sns-pldi}
is akin to that of \citet{Wang2012}.
In contrast, our approach is also dynamic but employs evaluation update rather
than recording traces.
Currently, our system allows the user to preview changes to the code and output,
but does not attempt to characterize and communicate the overall impact of the
changes.

\parahead{Programming with Direct Manipulation in Prior \sns{}}

\citet{sns-pldi} and \citet{sns-uist} developed a
\emph{direct manipulation programming system}
for generating and manipulating SVG graphic designs.
\citet{sns-uist} propose that graphical user interface features should
be co-designed with program transformations that aim to make
``large,'' structural, and often semantics-changing edits that codify the user
actions.
In future work, it would be useful to develop analogous ``code builders'' for our HTML
setting.

More closely related to our work is the approach of
\citet{sns-pldi}, which allows ``small'' changes
to output values to be reconciled with the program.
Their approach records \emph{value traces} for all numeric values.
When the user updates a number, the corresponding
value-trace equation is immediately solved, applied to the program, and the
new output is rendered---the resulting workflow provides a continuous, ``live''
interaction for equations that can be solved in almost real-time.
When there are multiple solutions, their approach employs simple
heuristics to automatically choose, favoring continuous updates over user
interaction to resolve intent.
The primary technical differences in our
evaluation update algorithm are that:
arbitrary types of values can be changed;
custom update behavior can be defined; and
time overhead (from re-evaluation) is traded to save space overhead (from
recording traces).
The tradeoff between time and space overhead
suggests that a hybrid, demand-driven approach may be worth
investigating, for large programs where both time
and memory are limited resources.

\section{Conclusion}

We presented \emph{bidirectional evaluation}, which allows
arbitrary programs in a general-purpose functional language to be run
``in reverse.''
When the output of a program is changed, bidirectional evaluation
synthesizes program repairs based on differences between the
original and modified output values.
We demonstrated the practicality of our approach by implementing it
within the \snsLeo{} direct manipulation programming system, using it
to develop a variety of HTML documents and applications that can be
interactively edited because of bidirectional evaluation.
We believe these techniques serve as a foundation for a variety of
systems to allow users to combine programming with direct manipulation.

\begin{acks}
This work was supported in part by
Swiss National Science Foundation Early Postdoc.Mobility Fellowship No.~175041,
European Research Council Project ``Implicit Programming'' GA 306484-IMPRO, and
\grantsponsor{GS100000001}{U.S. National Science Foundation}{http://dx.doi.org/10.13039/100000001}
Grant No.~\grantnum{GS100000001}{1651794}.
The authors would like to thank Justin Lubin for significant contributions to the \sns{}
implementation,
and Nate Foster, Anders Miltner, and Benjamin Pierce for comments about related work.

\end{acks}

%% \bibliography{references}

\clearpage

\appendix

\section{Appendix}
\label{sec:proofs}

The theorems below---\textsc{EvalUpdate} and \textsc{Conservative
UpdateEval}---pertain to the ``base'' bidirectional evaluation
formulation presented in \autoref{sec:leo-basic-update}
(\autoref{fig:eval-update}, \autoref{fig:update-numbers}, and
\autoref{fig:update-cons}).
This system does not support lenses, as presented in
\autoref{sec:leo-custom-update}, because we do not impose any
requirements on user-defined \verb+update+ functions.

\begin{theorem}[EvalUpdate] \label{prop:eval-update}
If $\reducesTo{\varEnv}{\varExp}{\varVal}$,
then $\updatesTo{\varEnv}{\varExp}{\varVal}{\varEnv}{\varExp}$.
\end{theorem}

\begin{proof}
Straightforward induction on the evaluation derivation.
\end{proof}

\begin{lemma} \label{prop:merge-equiv}
If $\varEnv' = \mergeEnvsTwoWay{\varEnv_1}{\varEnv_2}{\varExp_1}{\varExp_2}$,
then $\equivEnvs{\varEnv'}{\varEnv_1}{\freeVars{\varExp_1}}$
and  $\equivEnvs{\varEnv'}{\varEnv_2}{\freeVars{\varExp_2}}$.
\end{lemma}

\begin{proof}
By definition.
\end{proof}

\begin{lemma} \label{prop:equiv-eval}
If $\equivEnvs{\varEnv}{\varEnv'}{\freeVars{\varExp}}$
and $\reducesTo{\varEnv}{\varExp}{\varVal_1}$
and $\reducesTo{\varEnv'}{\varExp}{\varVal_2}$,
then $\varVal_1 = \varVal_2$.
\end{lemma}

\begin{proof}
Straightforward induction on the evaluation derivation.
\end{proof}

\begin{theorem}[Conservative UpdateEval] \label{prop:weak-correctness}
If $\updatesToPessimistic{\varEnv}{\varExp}{\varVal'}{\varEnv'}{\varExp'}$,
then $\reducesTo{\varEnv'}{\varExp'}{\varVal'}$.
\end{theorem}

\begin{proof}

By induction on the update derivation.
As discussed in \autoref{sec:leo-basic-update}, there are two challenges.
The first is when there are variable conflicts (not all uses of a variable have
been updated consistently).
This situation may arise in any rule that updates multiple
subexpressions; in each case, the conservative two-way environment
merge validates the conditions required for correctness.
Below, we show the \ruleName{U-Let} case as a representative.

\begin{center}
%
%% Case:
%
{\small
$\inferrule*[right=\ruleNameFig{U-Let}]
  {
  {
  \startAlignedPremises{}
   \reducesToAligned
     {\varEnv}
     {\varExp_1}
     {\varVal_1}
   \\
   \updatesToAlignedPessimistic
    {\envCat{\varEnv}{\envBind{\varVar}{\varVal_1}}}
    {\varExp_2}
    {\varVal_2'}
    {\envCat{\varEnv_2}{\envBind{\varVar}{\varVal_1'}}}
    {\varExp_2'}
  \stopAlignedPremises{}
  }
  \sepPremise
  {
  \startAlignedPremises{}
   \updatesToAlignedPessimistic
    {\varEnv}
    {\varExp_1}
    {\varVal_1'}
    {\varEnv_1}
    {\varExp_1'}
   %% \\
   %% \phantom{\ }
   \\
   \equalsAligned{\varEnv'}{\mergeEnvsTwoWay{\varEnv_1}{\varEnv_2}{\varExp_1}{\varExp_2}}
   \phantom{\overset{\pessimisticMark}{\rightsquigarrow}}
  \stopAlignedPremises{}
  }
  }
  {\updatesToPessimistic
    {\varEnv}
    {\expLet{\varVar}{\varExp_1}{\varExp_2}}
    {\varVal_2'}
    {\varEnv'}
    {\expLet{\varVar}{\varExp_1'}{\varExp_2'}}
  }
$
}
\end{center}

By the induction hypothesis,
$\reducesTo{\varEnv_1}{\varExp_1'}{\varVal_1'}$ and
$\reducesTo{(\envCat{\varEnv_2}{\envBind{\varVar}{\varVal_1'}})}{\varExp_2'}{\varVal_2'}$

By Lemma~\ref{prop:merge-equiv},
$\equivEnvs{\varEnv'}{\varEnv_1}{\freeVars{\varExp_1}}$ and
$\equivEnvs{\varEnv'}{\varEnv_2}{\freeVars{\varExp_2}}$.

By Lemma~\ref{prop:equiv-eval},
$\reducesTo{\varEnv'}{\varExp_1'}{\varVal_1'}$.

Because $\equivEnvs{\varEnv'}{\varEnv_2}{\freeVars{\varExp_2}}$,
by definition,
$\equivEnvs{(\envCat{\varEnv'}{\envBind{\varVar}{\varVal_1'}})}
           {(\envCat{\varEnv_2}{\envBind{\varVar}{\varVal_1'}})}
           {\freeVars{\varExp_2}}$.

By Lemma~\ref{prop:equiv-eval},
$\reducesTo{(\envCat{\varEnv'}{\envBind{\varVar}{\varVal_1'}})}{\varExp_2'}{\varVal_2'}$.

By \ruleName{E-Let},
$\reducesTo{\varEnv'}{\expLet{\varVar}{\varExp_1'}{\varExp_2'}}{\varVal_2'}$,
which is the goal.
\\

The second challenge is that if-expressions may not take the same branch.
Below, we show the \ruleName{U-If-True} case.
Again, two-way merge performs the additional checks required to guarantee
that the updated program evaluates to the new value.

\begin{center}
%
%% Case:
%
{\small
$\inferrule*[right=\ruleNameFig{U-If-True}]
  {
  {
  \startAlignedPremises{}
     \reducesToAligned{\varEnv}{\varExp_1}{\expTrue} \\
     \updatesToAlignedPessimistic{\varEnv}{\varExp_2}{\varVal'}{\varEnv_2}{\varExp_2'}
  \stopAlignedPremises{}
  }
  \\
  \varEnv' = \mergeEnvsTwoWay{\varEnv}{\varEnv_2}{\varExp_1}{\varExp_2}
  }
  {\updatesToPessimistic
    {\varEnv}
    {\expIte{\varExp_1}{\varExp_2}{\varExp_3}}
    {\varVal'}
    {\varEnv'}
    {\expIte{\varExp_1}{\varExp_2'}{\varExp_3}}
  }
$
}
\end{center}

By the induction hypothesis,
$\reducesTo{\varEnv_2}{\varExp_2'}{\varVal_2'}$.

By Lemma~\ref{prop:merge-equiv},
$\equivEnvs{\varEnv'}{\varEnv}{\freeVars{\varExp_1}}$.

By Lemma~\ref{prop:equiv-eval},
$\reducesTo{\varEnv'}{\varExp_1}{\expTrue}$.

By \ruleName{E-If-True},
$\reducesTo{\varEnv'}{\expIte{\varExp_1}{\varExp_2'}{\varExp_3}}{\varVal_2'}$,
which is the goal.

\end{proof}

\end{document}